\documentclass[11pt]{article}%
\usepackage{amsfonts}
\usepackage{amsmath}
\usepackage{amssymb}
\usepackage{graphicx}%
\providecommand{\U}[1]{\protect\rule{.1in}{.1in}}
\newtheorem{theorem}{Theorem}
\newtheorem{acknowledgement}[theorem]{Acknowledgement}

\newtheorem{corollary}[theorem]{Corollary}

\newtheorem{lemma}[theorem]{Lemma}
\newtheorem{notation}[theorem]{Notation}

\newtheorem{proposition}[theorem]{Proposition}
\newtheorem{remark}[theorem]{Remark}

\newenvironment{proof}[1][Proof]{\noindent\textbf{#1.} }{\ \rule{0.5em}{0.5em}}
\begin{document}

\title{Partial Traces and the Geometry of Entanglement; Sufficient Conditions for the
Separability of Gaussian States}
\author{Nuno Costa Dias\thanks{ncdias@meo.pt}\\University of Lisbon (GFM)
\and Maurice de Gosson\thanks{maurice.de.gosson@univie.ac.at}\\University of Vienna\\Faculty of Mathematics (NuHAG)
\and Jo\~{a}o Nuno Prata\thanks{joao.prata@mail.telepac.pt }\\University of Lisbon (GFM)}
\maketitle

\begin{abstract}
The notion of partial trace of a density operator is essential for the
understanding of the entanglement and separability properties of quantum
states. In this paper we investigate these notions putting an emphasis on the
geometrical properties of the covariance ellipsoids of the reduced states. We
thereafter focus on Gaussian states and we give new and easily numerically
implementable sufficient conditions for the separability of all Gaussian
states. Unlike the positive partial transposition criterion, none of these
conditions is however necessary.
\end{abstract}
\tableofcontents

\section*{Introduction}

Mixed quantum states play a pivotal role in quantum mechanics and its
applications (for instance teleportation, cryptography, quantum computation
and optics, to name a few). Mixed states are identified for all practical
purposes with their density operators (or matrices), which are positive
semidefinite self-adjoint operators with trace one on a Hilbert space
$\mathcal{H}$. One of the most important problems in density operator theory,
which is still largely open at the time being, is the characterization of the
separability of density operators or, which amounts to the same, of the
entanglement properties of mixed quantum states. In the case $\mathcal{H}%
=L^{2}(\mathbb{R}^{n})$ (which we assume from now on) necessary conditions for
separability can be found in the literature; one of the oldest is the
Peres--Horodecki criterion \cite{HHH1,Peres} on the partial transpose of a
density operator; more recently Werner and Wolf \cite{ww1} have proposed a
geometric condition involving the covariance matrix of the state. This
condition is also sufficient for separability for all density operators with
Wigner distribution
\begin{equation}
\rho(z)=\frac{1}{(2\pi)^{n}\sqrt{\det\Sigma}}e^{-\frac{1}{2}\Sigma^{-1}z^{2}}
\label{rhoG}%
\end{equation}
the covariance matrix $\Sigma$ being subjected to the quantum condition
\begin{equation}
\Sigma+\frac{i\hbar}{2}J\geq0 \label{quant0}%
\end{equation}
(see \S \ref{subsec11} for a discussion of this condition). It requires that
the covariance ellipsoid
\[
\Omega=\{z:\tfrac{1}{2}\Sigma^{-1}z^{2}\leq1\}
\]
has symplectic capacity at least $\pi\hbar$; this property, which is a
topological formulation of the uncertainty principle, means that there exists
a symplectic automorphisms of $\mathbb{R}^{2n}$ sending the phase space ball
with radius $\sqrt{\hbar}$ inside $\Omega$.

We will discuss the partial traces $\widehat{\rho}_{A}$ and $\widehat{\rho
}_{B}$ of a density operator $\widehat{\rho}$ with respect to a splitting
$\mathbb{R}^{2n}\equiv\mathbb{R}^{2n_{A}}\oplus\mathbb{R}^{2n_{B}}$ of phase
space, and show that the covariance ellipsoids of $\widehat{\rho}_{A}$ and
$\widehat{\rho}_{B}$ are the orthogonal projections of $\Omega$ onto the
reduced phase spaces $\mathbb{R}^{2n_{A}}$ and $\mathbb{R}^{2n_{B}}$. We will
see that if in particular $\widehat{\rho}$ is a Gaussian then these reduced
states are themselves Gaussian states with Wigner distributions%
\begin{align*}
\rho_{A}(z_{A})  &  =\frac{1}{(2\pi)^{n_{A}}\sqrt{\det\Sigma_{AA}}}%
e^{-\frac{1}{2}\Sigma_{AA}^{-1}z_{A}^{2}}\\
\rho_{B}(z_{B})  &  =\frac{1}{(2\pi)^{n_{B}}\sqrt{\det\Sigma_{BB}}}%
e^{-\frac{1}{2}\Sigma_{BB}^{-1}z_{B}^{2}}%
\end{align*}
where the reduced covariance matrices $\Sigma_{AA}$ and $\Sigma_{BB}$ are
calculated from the total covariance matrix $\Sigma$ using the theory of Schur
complements (see \S \ref{subsec22}), and the corresponding covariance
ellipsoids
\begin{align*}
\Omega_{A}  &  =\{z_{A}\in\mathbb{R}^{2n_{A}}:\tfrac{1}{2}\Sigma_{AA}%
^{-1}z_{A}^{2}\leq1\}\\
\Omega_{B}  &  =\{z_{B}\in\mathbb{R}^{2n_{B}}:\tfrac{1}{2}\Sigma_{BB}%
^{-1}z_{B}^{2}\leq1\}
\end{align*}
are the orthogonal projections (or \textquotedblleft shadows\textquotedblright%
) of the covariance ellipsoid $\Omega$ on the reduced phase spaces
$\mathbb{R}^{2n_{A}}$ and $\mathbb{R}^{2n_{B}}$, respectively.

The main new results are stated and proved in Sections \ref{Section3} and
\ref{secsuff1}. In these sections we discuss the separability of Gaussian
states. In Section \ref{Section3}, we prove a necessary and sufficient
condition for the separability of Gaussian states, which amounts to a
refinement of the Werner-Wolf condition. In Section \ref{secsuff1} we prove
various sufficient conditions for the separability of Gaussian states, and
show that, while sufficient, they are not necessary conditions.

\begin{notation}
The standard symplectic form on $\mathbb{R}^{n}\times\mathbb{R}^{n}$ is
$\sigma=\sum_{j=1}^{n}dp_{j}\wedge dx_{j}$; in matrix notation $\sigma
(z,z^{\prime})=Jz\cdot z^{\prime}=(z^{\prime})^{T}Jz$ where $J=%
\begin{pmatrix}
0_{n\times n} & I_{n\times n}\\
-I_{n\times n} & 0_{n\times n}%
\end{pmatrix}
$ and $\cdot$ denotes the Euclidean scalar product. We denote by
$\operatorname*{Sp}(n)$ the symplectic group of $(\mathbb{R}^{2n},\sigma)$.
Given a tempered distribution $a\in\mathcal{S}^{\prime}(\mathbb{R}^{2n})$ we
denote by $\operatorname*{Op}\nolimits_{\mathrm{W}}(a)$ the Weyl operator with
symbol $a$. The metaplectic group $\operatorname*{Mp}(n)$ is a faithful
unitary representation of the double cover of $\operatorname*{Sp}(n)$;
elements of $\operatorname*{Mp}(n)$ are denoted by $\widehat{S}$ and their
projections on $\operatorname*{Sp}(n)$ by $S$. Given $S \in Sp(n)$, and $R>0$,
the symplectic ball $S\left(  B^{2n}(R)\right)  $ is the ellipsoid:
\[
S\left(  B^{2n}(R)\right)  = \left\{  Sz: ~ |z| \leq R \right\}  ~.
\]

\end{notation}

\section{Partial Traces and Reduced States\label{sec1}}

\subsection{Density operators: basics\label{subsec11}}

Let $\widehat{\rho}\in\mathcal{L}_{1}(L^{2}(\mathbb{R}^{n}))$ be a positive
semidefinite operator with trace $\operatorname*{Tr}(\widehat{\rho})=1$ on
$L^{2}(\mathbb{R}^{n})$. In particular $\widehat{\rho}$ is self-adjoint and
compact. Such operators represent the mixed states of quantum mechanics and we
will freely identify them with these states. It follows from the spectral
theorem that there exists a sequence $(\lambda_{j})_{j\in\mathcal{I}}$
($\mathcal{I}$ a discrete index set) of nonnegative real numbers with
$\sum_{j\in\mathcal{I}}\lambda_{j}=1$ and an orthonormal basis $(\psi
_{j})_{j\in\mathcal{I}}$ of $L^{2}(\mathbb{R}^{n})$ such that $\widehat{\rho
}=\sum_{j\in\mathcal{I}}\lambda_{j}\widehat{\Pi}_{j}$ where $\widehat{\Pi}%
_{j}$ is the orthogonal projection on the ray $\mathbb{C}\psi_{j}$. The number%
\begin{equation}
\mu(\widehat{\rho})=\sum_{j\in\mathcal{I}}\lambda_{j}^{2}=\operatorname*{Tr}%
(\widehat{\rho}^{2}) \label{pur}%
\end{equation}
is called the purity of $\widehat{\rho}$ and we have $\mu(\widehat{\rho})=1$
if and only if one of the coefficients $\lambda_{j}$\ is equal to one, in
which case $\widehat{\rho}=\widehat{\Pi}_{j}$ is called a pure state. Density
operators are Weyl operators in their own right; in fact $\widehat{\rho}%
=(2\pi\hbar)^{n}\operatorname*{Op}\nolimits_{\mathrm{W}}(\rho)$ where
\begin{equation}
\rho=\sum_{j\in\mathcal{I}}\lambda_{j}W\psi_{j} \label{Wigner}%
\end{equation}
the $W\psi_{j}\in L^{2}(\mathbb{R}^{2n})$ being the Wigner transforms of the
functions $\psi_{j}$; it follows from Moyal's identity \cite{Wigner} that the
$W\psi_{j}$ form an orthonormal subset of $L^{2}(\mathbb{R}^{2n})$. The
operator $\widehat{\rho}$ is the bounded operator on $L^{2}(\mathbb{R}^{n})$
with square-integrable distributional kernel
\begin{equation}
K(x,y)=\int_{\mathbb{R}^{n}}e^{\frac{i}{\hbar}p(x-y)}\rho(\tfrac{1}{2}(x+y),p)
dp~. \label{K}%
\end{equation}
It is current practice in the physically oriented literature to write
\begin{equation}
\operatorname*{Tr}(\widehat{\rho})=\int_{\mathbb{R}^{n}}K(x,x)dx \label{trK}%
\end{equation}
which leads, setting $x=y$ in (\ref{K}), to
\begin{equation}
\operatorname*{Tr}(\widehat{\rho})=\int_{\mathbb{R}^{2n}}\rho(z)dz=1~.
\label{troh}%
\end{equation}
One has however to view these formulas with a more than critical eye; they are
generally false unless some additional conditions are imposed on $\rho(z)$
(see \cite{Birkbis,QUANTA} and the references therein). Formula (\ref{troh})
however holds true if one makes the extra assumption that $\rho\in
L^{1}(\mathbb{R}^{2n})$ (see \cite{duwong}). We will use in this paper the
following stronger result due to Shubin (\cite{sh87}, \S 27. Setting
\[
\langle z\rangle=(1+|z|^{2})^{1/2}%
\]
for $z\in\mathbb{R}^{2n}$ we have:

\begin{proposition}
[Shubin]\label{Shubinprop}Let $\widehat{\rho}$ be a bounded operator with Weyl
symbol $(2\pi\hbar)^{n}\rho$ If $\rho\in C^{\infty}(\mathbb{R}^{2n})$ and all
its $z$-derivatives $\partial_{z}^{\alpha}\rho$ satisfy estimates%
\begin{equation}
|\partial_{z}^{\alpha}\rho(z)|\leq C_{\alpha}\langle z\rangle^{m-|\alpha|}
\label{shubin}%
\end{equation}
with $m<-2n$ and $C_{\alpha}>0$, then the operator $\widehat{\rho}$ is of
trace class and we have%
\begin{equation}
\operatorname*{Tr}(\widehat{\rho})=\int_{\mathbb{R}^{2n}}\rho(z)dz~.
\label{trace}%
\end{equation}

\end{proposition}

The interest of this result comes from the fact that one does not have to
assume from the beginning that $\widehat{\rho}$ is of trace class, let alone a
density operator. Notice that the trace formula (\ref{trace}) automatically
follows since the condition (\ref{shubin}) implies that\ $\rho\in
L^{1}(\mathbb{R}^{2n})$.

We will denote by $\Gamma^{m}(\mathbb{R}^{2n})$ the Shubin class of all
functions $\rho\in C^{\infty}(\mathbb{R}^{2n})$ satisfying the estimates
(\ref{shubin}) for all $\alpha\in\mathbb{N}^{n}$.

Let $\widehat{\rho}=(2\pi\hbar)^{n}\operatorname*{Op}\nolimits_{\mathrm{W}%
}(\rho)$ be a density operator. We assume that
\begin{equation}
\int_{\mathbb{R}^{2n}}\langle z\rangle^{2}|\rho(z)|dz<\infty~; \label{conv}%
\end{equation}
this ensures us of the existence of first and second order momenta. This
condition holds if for instance $\rho$ belongs to some Shubin symbol class
$\Gamma^{m}(\mathbb{R}^{2n})$ with $m<-2n-2$. Let $\alpha,\beta=1,...,2n$ and
$z_{\alpha}=x_{\alpha}$ for $1\leq\alpha\leq n$ and $z_{\alpha}=p_{\alpha}$
for $n+1\leq\alpha\leq2n$. The average value of $\widehat{\rho}$ is defined by
$\bar{z}=(\bar{z}_{1},...,\bar{z}_{2n})$ where
\begin{equation}
\bar{z}_{\alpha}=\int_{\mathbb{R}^{2n}}z_{\alpha}\rho(z)dz \label{zalpha}%
\end{equation}
and the covariances\ are given by the integrals\
\begin{equation}
\Sigma(z_{\alpha},z_{\beta})=\int_{\mathbb{R}^{2n}}(z_{\alpha}-\bar{z}%
_{\alpha})(z_{\beta}-\bar{z}_{\beta})\rho(z)dz~. \label{cov}%
\end{equation}
The covariance matrix of $\widehat{\rho}$ is, by definition, the $2n\times2n$
matrix%
\begin{equation}
\Sigma=(\Sigma(z_{\alpha},z_{\beta}))_{1\leq\alpha,\beta\leq2n}
\label{covmatrix}%
\end{equation}
or, in more compact form,%
\[
\Sigma=\int_{\mathbb{R}^{2n}}(z-\bar{z})(z-\bar{z})^{T}\rho(z)dz
\]
where $z$ is viewed as a column vector $\binom{x}{p}$. The condition
$\widehat{\rho}\geq0$ requires that \cite{Narcow89,Narcow90,Werner}
\begin{equation}
\Sigma+\frac{i\hbar}{2}J\geq0 \label{quant}%
\end{equation}
where \textquotedblleft$\geq0$\textquotedblright\ means \textquotedblleft is
positive semidefinite\textquotedblright\ (note that all the eigenvalues of
$\Sigma+\frac{i\hbar}{2}J$ are real since it is a self-adjoint matrix). This
condition implies, in particular, that $\Sigma>0$; it is actually an
equivalent form of the Robertson--Schr\"{o}dinger inequalities
\cite{FOOP,goluPR}. It is a symplectically invariant formulation of the
uncertainty principle of quantum mechanics: introducing the covariance
ellipsoid
\begin{equation}
\Omega=\{z\in\mathbb{R}^{2n}:\tfrac{1}{2}\Sigma^{-1}z^{2}\leq1\}
\label{covellipse}%
\end{equation}
condition (\ref{quant}) can be rewritten as%
\begin{equation}
c(\Omega)\geq\pi\hbar\label{quantcamel}%
\end{equation}
where $c(\Omega)$ is the symplectic capacity of the ellipsoid $\Omega$
\cite{Birk,FOOP,Birkbis,goluPR}. Equivalently:%

\begin{equation}
\text{\emph{There exists} }S\in\operatorname*{Sp}(n)\text{ \emph{such that}
}SB^{2n}(\sqrt{\hbar})\subset\Omega~. \label{quantcamelbis}%
\end{equation}
The symplectic balls $SB^{2n}(\sqrt{\hbar})$ are minimum uncertainty
ellipsoids; it is convenient to use the following terminology
\cite{blob,goluPR} as it simplifies many statements:
\begin{gather}
\text{\emph{A quantum blob in} }\mathbb{R}^{2n}\text{ \emph{is a symplectic
ball} }\label{blob}\\
S(B^{2n}(R))\text{ \emph{with radius} }R=\sqrt{\hbar}~.\nonumber
\end{gather}

These properties all follow from the following observation:

\begin{proposition}
\label{propWill}Let $\lambda_{1,\sigma},...,\lambda_{n,\sigma}$ be the
symplectic eigenvalues of $\Sigma$, that is, $\lambda_{j,\sigma}>0$ and $\pm
i\lambda_{j,\sigma}$ is an eigenvalue of $J\Sigma$ for all $j=1,...,n$. The
condition $\Sigma+\frac{i\hbar}{2}J\geq0$ is equivalent to the conditions
$\lambda_{j,\sigma}\geq\frac{1}{2}\hbar$ for all $j=1,...,n$.
\end{proposition}

\begin{proof}
See \cite{Birk,FOOP}. It is based on the use of Williamson's symplectic
diagonalization theorem: $M$ being positive definite there exists
$S\in\operatorname*{Sp}(n)$ such that
\begin{equation}
M=S^{T}DS\text{ \ },\text{ \ }D=%
\begin{pmatrix}
\Lambda & 0\\
0 & \Lambda
\end{pmatrix}
\text{ , }\Lambda=\operatorname*{diag}(\lambda_{1,\sigma},...,\lambda
_{n,\sigma}) \label{Williamson}%
\end{equation}
(see for instance \cite{Folland,Giedke,Birk}). Notice that the eigenvalues of
$J\Sigma$ are those of the antisymmetric matrix $\Sigma^{1/2}J\Sigma^{1/2}$
and are hence indeed of the type $\pm i\lambda$ with $\lambda>0$.
\end{proof}

\subsection{Reduced density operators\label{subsec22}}

Let $n_{A}$, $n_{B}$ be two integers such that $n=n_{A}+n_{B}$. We identify
the direct sum $\mathbb{R}^{2n_{A}}\oplus\mathbb{R}^{2n_{B}}$ with
$\mathbb{R}^{2n}$ and the symplectic form $\sigma$ with $\sigma_{A}%
\oplus\sigma_{B}$ where $\sigma_{A}$ (\textit{resp}. $\sigma_{B}$) is the
standard symplectic form on $\mathbb{R}^{2n_{A}}$ (\textit{resp}.
$\mathbb{R}^{2n_{B}}$).

Let $\widehat{\rho}$ be a density operator on $L^{2}(\mathbb{R}^{n})$ with
Wigner distribution $\rho$. Assuming that $\rho$ satisfies the Shubin
estimates (\ref{shubin}) for some $m<-2n$, we define the \textit{reduced
density operator} $\widehat{\rho}_{A}$ by the formula%
\begin{equation}
\widehat{\rho}_{A}=(2\pi\hbar)^{n_{A}}\operatorname*{Op}\nolimits_{\mathrm{W}%
}(\rho_{A}) \label{reducedA}%
\end{equation}
where we have set%
\begin{equation}
\rho_{A}(z_{A})=\int_{\mathbb{R}^{n_{B}}}\rho(z_{A},z_{B})dz_{B}~.
\label{roadef}%
\end{equation}

This terminology has of course to be justified (it is not \textit{a priori}
clear why $\widehat{\rho}_{A}$ should be a density operator). Let us recall
the following result from quantum harmonic analysis which reduces to a
classical theorem of Bochner \cite{Bochner} on functions of positive type when
$\hbar=0$:

\begin{proposition}
[KLM conditions]\label{PropKLM}Let $a\in L^{1}(\mathbb{R}^{2n})$ and assume
that $\widehat{A}=\operatorname*{Op}\nolimits_{\mathrm{W}}(a)$ is of trace
class. We have $\widehat{A}\geq0$ if and only if the symplectic Fourier
transform $a_{\Diamond}=F_{\Diamond}a$ defined by
\begin{equation}
a_{\Diamond}(z)=\int_{\mathbb{R}^{2n}}e^{i\sigma(z,z^{\prime})}a(z^{\prime
})dz^{\prime} \label{diam}%
\end{equation}
is continuous\footnote{With the assumption $a\in L^{1}(\mathbb{R}^{2n})$, this
is automatically valid, via the Riemann-Lebesgue lemma.} and of $\hbar
$-positive type, that is if for every integer $N$ the $N\times N$ matrix
$\Lambda_{(N)}$ with entries%
\begin{equation}
\Lambda_{jk}=e^{-\frac{i\hbar}{2}\sigma(z_{j},z_{k})}a_{\Diamond}(z_{j}-z_{k})
\label{lajk}%
\end{equation}
is positive semidefinite for all choices of $(z_{1},z_{2},...,z_{N}%
)\in(\mathbb{R}^{2n})^{N}$.
\end{proposition}

The proof of this result goes back to the seminal work of Kastler
\cite{Kastler} and Loupias and Miracle-Sole \cite{LouMiracle1,LouMiracle2}.
While these authors use the theory of $C^{\ast}$-algebras and hard functional
analysis, one of us has recently given in \cite{cogoni} a conceptually simpler
proof using the properties of the Heisenberg--Weyl displacement operators
$\widehat{T}(z)=e^{-i\sigma(\widehat{z},z)/\hbar}$ \cite{Birk,Birkbis}.

\begin{proposition}
\label{mA}Let $\rho\in\Gamma^{m}(\mathbb{R}^{2n})$ for some $m<-2n$. The
operator $\widehat{\rho}_{A}=(2\pi\hbar)^{n_{A}}\operatorname*{Op}%
\nolimits_{\mathrm{W}}(\rho_{A})$ is a density operator on $L^{2}%
(\mathbb{R}^{n_{A}})$ and we have $\rho_{A}\in\Gamma^{m_{A}}(\mathbb{R}%
^{2n_{A}})$ for every $m_{A}<-2n_{A}$.
\end{proposition}

\begin{proof}
The integral (\ref{roadef}) is convergent in view of the trivial inequality
$(1+|z|^{2})^{m}\leq(1+|z_{B}|^{2})^{m}$. Choosing $m_{A}<-2n_{A}$ and
$m_{B}<-2n_{B}$ such that $m=m_{A}+m_{B}$ we have $\langle z\rangle
^{m-|\alpha|}\leq\langle z_{A}\rangle^{m_{A}-|\alpha|}\langle z_{B}%
\rangle^{m_{B}}$ as follows from the inequality
\[
(1+|z|^{2})^{m-|\alpha|}\leq(1+|z_{A}|^{2})^{m_{A}-|\alpha|}(1+|z_{B}%
|^{2})^{m_{B}}~.
\]
Using the Shubin estimates (\ref{shubin}) we thus have
\begin{align*}
\partial_{z_{A}}^{\alpha}\rho_{A}(z_{A})  &  =\int_{\mathbb{R}^{2n_{B}}%
}\partial_{z_{A}}^{\alpha}\rho(z_{A},z_{B})dz_{B}\\
&  \leq C_{\alpha}\langle z\rangle^{m-|\alpha|}\int_{\mathbb{R}^{2n_{B}}%
}\langle z_{B}\rangle^{m_{B}}dz_{B}%
\end{align*}
and hence $\rho_{A}\in\Gamma^{m_{A}}(\mathbb{R}^{2n_{A}})$ since the integral
over $\mathbb{R}^{n_{B}}$ is convergent in view of the inequality
$m_{B}<-2n_{B}$. It follows from Proposition \ref{Shubinprop} that
$\widehat{\rho}_{A}$ is a trace class operator whose trace is%
\begin{equation}
\operatorname*{Tr}\nolimits_{A}(\widehat{\rho}_{A})=\int_{\mathbb{R}^{2n_{A}}%
}\rho_{A}(z_{A})dz_{A}=1~. \label{trapa}%
\end{equation}
There remains to show that $\widehat{\rho}_{A}\geq0$ (and hence $\widehat{\rho
}_{A}^{\ast}=\widehat{\rho}_{A}$). In view of the KLM conditions (Proposition
\ref{PropKLM}) it is sufficient to prove that the Fourier transform $(\rho
_{A})_{\Diamond}$ is continuous and satisfies $\Lambda_{(N)}^{A}\geq0$ for
every integer $N>0$ where $\Lambda_{(N)}^{A}=(\Lambda_{jk}^{A})_{j,k}$ with
\[
\Lambda_{jk}^{A}=e^{-\frac{i\hbar}{2}\sigma_{A}(z_{A,j},z_{A,k})}(\rho
_{A})_{\Diamond}(z_{A,j}-z_{A,k})
\]
the vectors $z_{A,j}$ and $z_{A,k}$ of $\mathbb{R}^{2n_{A}}$ being arbitrary.
The continuity of $(\rho_{A})_{\Diamond}$ being obvious (Riemann--Lebesgue
Lemma) all we have to do is to show that $\Lambda_{(N)}^{A}\geq0$. We first
observe that by Fubini's theorem $(\rho_{A})_{\Diamond}(z_{A})=\rho_{\Diamond
}(z_{A}\oplus0)$ and hence
\[
\Lambda_{jk}^{A}=e^{-\frac{i\hbar}{2}\sigma(z_{A,j}\oplus0,z_{A,k}\oplus
0)}\rho_{\Diamond}((z_{A,j}\oplus0)-(z_{A,k}\oplus0))~;
\]
the matrix $\Lambda_{(N)}^{A}$ is thus the matrix $\Lambda_{(N)}$
corresponding to the particular choices $z_{j}=z_{A,j}\oplus0$ and
$z_{k}=z_{A,k}\oplus0$. Since $\rho_{\Diamond}$ satisfies the KLM conditions
we must have $\Lambda_{(N)}^{A}\geq0$, hence $(\rho_{A})_{\Diamond}$ also
satisfies them.
\end{proof}

From now on we will write the covariance matrix $\Sigma$ in the $AB$-ordering
as
\begin{equation}
\Sigma=%
\begin{pmatrix}
\Sigma_{AA} & \Sigma_{AB}\\
\Sigma_{BA} & \Sigma_{BB}%
\end{pmatrix}
\text{ \textit{with} }\Sigma_{BA}=\Sigma_{AB}^{T} \label{sigmab}%
\end{equation}
the blocks $\Sigma_{AA}$, $\Sigma_{AB}$, $\Sigma_{BA}$, $\Sigma_{BB}$ having
dimensions $2n_{A}\times2n_{A}$, $2n_{A}\times2n_{B}$, $2n_{B}\times2n_{A}$,
$2n_{B}\times2n_{B}$, respectively. In this notation the quantum condition
(\ref{quant}) reads
\begin{equation}
\Sigma+\frac{i\hbar}{2}J_{AB}\geq0~, \text{ \textit{with} } J_{AB}=%
\begin{pmatrix}
J_{A} & 0\\
0 & J_{B}%
\end{pmatrix}
~. \label{quant2}%
\end{equation}

The covariance matrices $\Sigma_{A}$ and $\Sigma_{B}$ of the reduced density
operators are, respectively, the blocks $\Sigma_{AA}$ and $\Sigma_{BB}$ of
$\Sigma$ as immediately follows from the definitions (\ref{zalpha}) and
(\ref{cov}) using the formulas
\[
\rho_{A}(z_{A})=\int_{\mathbb{R}^{n_{B}}}\rho(z_{A},z_{B})dz_{B}\text{ \ ,
\ }\rho_{B}(z_{B})=\int_{\mathbb{R}^{n_{A}}}\rho(z_{A},z_{B})dz_{A}~.
\]
These matrices satisfy the quantum conditions%
\begin{equation}
\Sigma_{AA}+\frac{i\hbar}{2}J_{A}\geq0\text{ and }\Sigma_{BB}+\frac{i\hbar}%
{2}J_{B}\geq0 \label{quantred}%
\end{equation}
and the covariance ellipsoids of $\widehat{\rho}_{A}$ and $\widehat{\rho}_{B}$
are%
\begin{equation}
\Omega_{A}=\{z_{A}:\tfrac{1}{2}\Sigma_{AA}^{-1}z_{A}^{2}\leq1\}\text{ and
}\Omega_{B}=\{z_{B}:\tfrac{1}{2}\Sigma_{BB}^{-1}z_{B}^{2}\leq1\}
\label{omegab}%
\end{equation}
(we will see below that they are just the orthogonal projections on
$\mathbb{R}^{2n_{A}}$ and $\mathbb{R}^{2n_{A}}$ of the covariance ellipsoid
$\Omega$). That the quantum conditions (\ref{quantred}) hold follows from the
fact that $\widehat{\rho}_{A}$ and $\widehat{\rho}_{B}$ are \textit{bona fide}
density operators, but this can also be seen directly by noting that
(\ref{quant2}) can be written%
\[%
\begin{pmatrix}
\Sigma_{AA}+\frac{i\hbar}{2}J_{A} & \Sigma_{AB}\\
\Sigma_{BA} & \Sigma_{BB}+\frac{i\hbar}{2}J_{B}%
\end{pmatrix}
\geq0~.
\]

The symmetric matrix
\begin{equation}
\Sigma/\Sigma_{BB}=\Sigma_{AA}-\Sigma_{AB}\Sigma_{BB}^{-1}\Sigma_{BA}~.
\label{Schur}%
\end{equation}
is called the \textit{Schur complement} \cite{horn,zhang} of the block
$\Sigma_{BB}$ of $\Sigma$. Using the obvious factorization%
\begin{equation}
\Sigma=%
\begin{pmatrix}
I_{A} & \Sigma_{AB}\Sigma_{BB}^{-1}\\
0 & I_{B}%
\end{pmatrix}%
\begin{pmatrix}
\Sigma/\Sigma_{BB} & 0\\
0 & \Sigma_{BB}%
\end{pmatrix}%
\begin{pmatrix}
I_{A} & 0\\
\Sigma_{BB}^{-1}\Sigma_{BA} & I_{B}%
\end{pmatrix}
\label{schurfact}%
\end{equation}
we readily get various formulas for the inverse of $\Sigma$; the one we will
use here is
\begin{equation}
\Sigma^{-1}=%
\begin{pmatrix}
(\Sigma/\Sigma_{BB})^{-1} & -(\Sigma/\Sigma_{BB})^{-1}\Sigma_{AB}\Sigma
_{BB}^{-1}\\
-\Sigma_{BB}^{-1}\Sigma_{BA}(\Sigma/\Sigma_{BB})^{-1} & (\Sigma/\Sigma
_{AA})^{-1}%
\end{pmatrix}
\label{covinv}%
\end{equation}
(see \cite{Tzon} for a review of various formulas for block-matrix inversion).
Also note that it immediately follows from (\ref{schurfact}) that%
\begin{equation}
\det\Sigma=\det(\Sigma/\Sigma_{BB})\det\Sigma_{BB}~. \label{schurdet}%
\end{equation}

\subsection{The shadows of the covariance ellipse\label{secshadow}}

In practice, we have to deal more often with the inverse of the covariance
matrix than with the covariance matrix itself (this occurred already above in
the definition of the covariance ellipsoid (\ref{covellipse})). It is
therefore useful to have an explicit formula for that inverse.

In particular, to study the orthogonal projections (\textquotedblleft
shadows\textquotedblright) of the covariance ellipsoid $\Omega$ on the reduced
phase spaces $\mathbb{R}^{2n_{A}}$ and $\mathbb{R}^{2n_{B}}$ it will be
convenient to set $M=\frac{\hbar}{2}\Sigma^{-1}$. We will write, using the
$AB$-ordering $z=(z_{A},z_{B})$,
\begin{equation}
M=%
\begin{pmatrix}
M_{AA} & M_{AB}\\
M_{BA} & M_{BB}%
\end{pmatrix}
\label{MMM}%
\end{equation}
where $M_{AA}$, $M_{AB}$, $M_{BA}$, $M_{BB}$ are, respectively, $2n_{A}%
\times2n_{A}$, $2n_{A}\times2n_{B}$, $2n_{B}\times2n_{A}$, $2n_{B}\times
2n_{B}$ matrices. In this notation the covariance ellipsoid of $\widehat{\rho
}$ is the set%
\begin{equation}
\Omega=\{z\in\mathbb{R}^{2n}:Mz^{2}\leq\hbar\} \label{couvm}%
\end{equation}
and the quantum condition $\Sigma+\frac{i\hbar}{2}J_{AB}\geq0$ becomes%
\[
M^{-1}+iJ_{AB}\geq0
\]
which is equivalent, in view of Proposition \ref{propWill}, to the statement:
\begin{equation}
\text{\textit{The symplectic eigenvalues of}}M\text{ \textit{are}}\leq1~.
\label{m1}%
\end{equation}

Notice that since $M$ is positive definite and symmetric (because $\Sigma$ is)
the blocks $M_{AA}$ and $M_{BB}$ are also symmetric and positive definite and
we have $M_{BA}=M_{AB}^{T}$.

The following general Lemma will be very useful in our geometric
considerations about separability:

\begin{lemma}
\label{lemmapim}Let $\Pi_{A}$ (resp. $\Pi_{B}$) be the orthogonal projection
$\mathbb{R}^{2n}\longrightarrow\mathbb{R}^{2n_{A}}$ (resp. $\mathbb{R}%
^{2n}\longrightarrow\mathbb{R}^{2n_{B}}$) and $\Omega_{R}$ the phase space
ellipsoid $\{z\in\mathbb{R}^{2n}:Mz^{2}\leq R^{2}\}$, for some $R>0$. We have
\begin{align}
\Pi_{A}\Omega_{R}  &  =\{z_{A}\in\mathbb{R}^{2n_{A}}:(M/M_{BB})z_{A}^{2}\leq
R^{2}\}\label{boundb}\\
\Pi_{B}\Omega_{R}  &  =\{z_{B}\in\mathbb{R}^{2n_{B}}:(M/M_{AA})z_{B}^{2}\leq
R^{2}\} \label{bounda}%
\end{align}
where
\begin{align}
M/M_{BB}  &  =M_{AA}-M_{AB}M_{BB}^{-1}M_{BA}\label{SchurA}\\
M/M_{AA}  &  =M_{BB}-M_{BA}M_{AA}^{-1}M_{AB} \label{SchurB}%
\end{align}
are the Schur complements.
\end{lemma}

\begin{proof}
Let us set $Q(z)=$ $Mz^{2}-R^{2}$; the boundary $\partial\Omega_{R}$ of the
hypersurface $Q(z)=0$ is defined by%
\begin{equation}
M_{AA}z_{A}^{2}+2M_{BA}z_{A}\cdot z_{B}+M_{BB}z_{B}^{2}=R^{2}~. \label{mabab}%
\end{equation}
A point $z_{A}$ belongs to $\partial\Pi_{A}\Omega_{R}$ if and only if the
normal vector to $\partial\Omega_{R}$ at the point $z=(z_{A},z_{B}) \in
\Omega_{R}$ is parallel to $\mathbb{R}^{2n_{A}}$, hence the constraint
$\partial_{z}Q(z)=2Mz\in\mathbb{R}^{2n_{A}}\oplus0$. This is equivalent to the
condition $M_{BA}z_{A}+M_{BB}z_{B}=0$, that is to $z_{B}=-M_{BB}^{-1}%
M_{BA}z_{A}$. Inserting $z_{B}$ in (\ref{mabab}) shows that the boundary
$\partial\Pi_{A}\Omega_{R}$ is the set $\Sigma_{A}=(M/M_{BB})z_{A}^{2}=R^{2}$
which yields (\ref{boundb}). Formula (\ref{bounda}) is proven in the same way.
\end{proof}

It follows from Lemma \ref{lemmapim} that the orthogonal projections on
$\mathbb{R}^{2n_{A}}$ and $\mathbb{R}^{2n_{B}}$ of the covariance ellipsoid
$\Omega$ of $\widehat{\rho}$ are just the covariance ellipsoids of the reduced
operators $\widehat{\rho}_{A}$ and $\widehat{\rho}_{B}$:

\begin{proposition}
The covariance ellipsoids $\Omega_{A}$ and $\Omega_{B}$ of the reduced quantum
states $\widehat{\rho}_{A}$ and $\widehat{\rho}_{B}$ are the orthogonal
projections on $\mathbb{R}^{2n_{A}}$ and $\mathbb{R}^{2n_{B}}$ of the
covariance ellipsoid $\Omega$ of $\widehat{\rho}$:%
\begin{align}
\Omega_{A}  &  =\Pi_{A}\Omega=\{z_{A}\in\mathbb{R}^{2n_{A}}:(M/M_{BB}%
)z_{A}^{2}\leq\hbar\}\label{pambb}\\
\Omega_{B}  &  =\Pi_{B}\Omega=\{z_{B}\in\mathbb{R}^{2n_{B}}:(M/M_{AA}%
)z_{B}^{2}\leq\hbar\}~. \label{pamaa}%
\end{align}

\end{proposition}

\begin{proof}
Let $M=\frac{\hbar}{2}\Sigma^{-1}$. Writing $M$ in block-matrix form
(\ref{MMM}), its inverse has the form
\begin{equation}
M^{-1}=%
\begin{pmatrix}
(M/M_{BB})^{-1} & \ast\\
\ast & (M/M_{AA})^{-1}%
\end{pmatrix}
\end{equation}
(\textit{cf.} formula (\ref{covinv})) and hence
\[
(M/M_{BB})^{-1}=\frac{\hbar}{2}\Sigma_{AA}\text{ \ and \ }(M/M_{AA}%
)^{-1}=\frac{\hbar}{2}\Sigma_{BB}\text{ }.
\]
Formulas (\ref{pambb}) and (\ref{pamaa}) follow using Lemma \ref{lemmapim}
with $R=\sqrt{\hbar}$.
\end{proof}

\section{The $AB$-separability of a density operator\label{sec2}}

In this section we study two necessary conditions for bipartite separability
of density operator on $L^{2}(\mathbb{R}^{n})$. The first (Proposition
\ref{Prop1}) is the so-called \textquotedblleft PPT
criterion\textquotedblright,\ of which we give a rigorous proof, and the
second (Proposition \ref{Prop2}) is a non-trivial refinement of a result due
to Werner and Wolf \cite{ww1}.

\subsection{The Peres--Horodecki condition}

We say that the operator $\widehat{\rho}$ is \textquotedblleft$AB$
separable\textquotedblright\ if there exist sequences of density operators
$\widehat{\rho}_{j}^{A}\in\mathcal{L}_{1}(L^{2}(\mathbb{R}^{n_{A}}))$ and
$\widehat{\rho}_{j}^{B}\in\mathcal{L}_{1}(L^{2}(\mathbb{R}^{n_{B}}))$ and real
numbers $\alpha_{j}\geq0$, $\sum_{j}\alpha_{j}=1$ such that
\begin{equation}
\widehat{\rho}=\sum_{j\in\mathcal{I}}\alpha_{j}\widehat{\rho}_{j}^{A}%
\otimes\widehat{\rho}_{j}^{B} \label{AB}%
\end{equation}
where the convergence is for the norm of $\mathcal{L}_{1}(L^{2}(\mathbb{R}%
^{n}))$).

Let us introduce some new notation. We denote by $I_{A}$ the identity
$(x_{A},p_{A})\longmapsto(x_{A},p_{A})$ and by $\overline{I}_{B}$ the
involution $(x_{B},p_{B})\longmapsto(x_{B},-p_{B})$. We set $\overline{I}%
_{AB}=I_{A}\oplus\overline{I}_{B}$ and, as before, $J_{AB}=J_{A}\oplus J_{B}$
where $J_{A}$ (\textit{resp.} $J_{B}$) is the standard symplectic matrix in
$\mathbb{R}^{2n_{A}}$ (\textit{resp}. $\mathbb{R}^{2n_{B}}$).

Given a general density operator $\widehat{\rho}=(2\pi\hbar)^{n}%
\operatorname*{Op}\nolimits_{\mathrm{W}}(\rho)$ there exists a necessary
condition for $AB$-separability; it is known in the physical literature as the
PPT criterion (PPT stands for \textquotedblleft positive partial
transpose\textquotedblright) and was first precisely stated in
\cite{HHH1,HHH2,Peres}. Also see the paper \cite{Simon} of Simon where it is
shown that the PPT criterion is sufficient for separability of Gaussian states
when $n_{A}=n_{B}=2$ (also see Duan \textit{et al}. \cite{Duan}). Below we
give a short and rigorous proof of this condition based on the (trivial)
equality%
\begin{equation}
W\psi(\overline{I}_{B}z_{B})=W\overline{\psi}(z_{B}) \label{trivial}%
\end{equation}
valid for all $\psi\in L^{2}(\mathbb{R}^{n_{B}})$.

\begin{proposition}
\label{Prop1}Let $\widehat{\rho}=(2\pi\hbar)^{n}\operatorname*{Op}%
\nolimits_{\mathrm{W}}(\rho)$ be a density operator on $\mathbb{R}%
^{2n}=\mathbb{R}^{2n_{A}}\oplus\mathbb{R}^{2n_{B}}$. Suppose that the
$AB$-separability condition
\begin{equation}
\widehat{\rho}=\sum_{j\in\mathcal{I}}\lambda_{j}\widehat{\rho}_{j}^{A}%
\otimes\widehat{\rho}_{j}^{B} \label{ABAB}%
\end{equation}
holds. Then the operator
\[
\widehat{\rho}^{T_{B}}=(2\pi\hbar)^{n}\operatorname*{Op}\nolimits_{\mathrm{W}%
}(\rho\circ\overline{I}_{AB})
\]
is also a density operator on $\mathbb{R}^{2n}=\mathbb{R}^{2n_{A}}%
\oplus\mathbb{R}^{2n_{B}}$.
\end{proposition}

\begin{proof}
Suppose that (\ref{ABAB}) holds; then $\rho=\sum_{j}\lambda_{j}\rho_{j}%
^{A}\otimes\rho_{j}^{B}$ and
\[
\rho_{j}^{A}=\sum_{\ell}\alpha_{j,\ell}W_{A}\psi_{j,\ell}^{A}\text{ \ ,
\ }\rho_{j}^{B}=\sum_{m}\beta_{j,m}W_{B}\psi_{j,m}^{B}%
\]
with $(\psi_{\ell}^{A},\psi_{\ell}^{B})\in L^{2}(\mathbb{R}^{n_{A}})\times
L^{2}(\mathbb{R}^{n_{B}})$ and $\alpha_{j,\ell},\beta_{j,m}\geq0$; that is
\[
\rho=\sum_{j,\ell,m}\gamma_{j,\ell,m}W_{A}\psi_{j,\ell}^{A}\otimes W_{B}%
\psi_{j,m}^{B}%
\]
where $\gamma_{j,\ell,m}=\lambda_{j}\alpha_{j,\ell}\beta_{j,m}\geq0$. We have
\[
\rho(\overline{I}_{AB}z)=\sum_{j\in\mathcal{I}}\lambda_{j}\rho_{j}^{A}%
(z_{A})\rho_{j}^{B}(\overline{I}_{B}z_{B});
\]
using (\ref{trivial}) we thus have%
\[
\rho\circ\overline{I}_{AB}=\sum_{j,\ell,m}\gamma_{j,\ell,m}W(\psi_{j,\ell}%
^{A}\otimes\overline{\psi}_{j,m}^{B})
\]
hence $\operatorname*{Op}\nolimits_{\mathrm{W}}(\rho\circ\overline{I}_{AB})$
is also a positive semidefinite trace class operator; that $\operatorname*{Tr}%
(\widehat{\rho}^{T_{B}})=\operatorname*{Tr}(\widehat{\rho})=1$ is obvious.
\end{proof}

Notice that we have $\widehat{\rho}^{T_{B}}=\sum_{j}\alpha_{j}\widehat{\rho
}_{j}^{A}\otimes(\widehat{\rho}_{j}^{B})^{T}$ where
\[
(\widehat{\rho}_{j}^{B})^{T}=(2\pi\hbar)^{n_{B}}\operatorname*{Op}%
\nolimits_{\mathrm{W}}(\rho_{j}\circ\overline{I}_{B})
\]
is the transpose of $\widehat{\rho}_{j}^{B}$, hence the denomination
\textquotedblleft partial positive transpose\textquotedblright\ for the
operator $\widehat{\rho}^{T_{B}}$ used in the literature.

Proposition \ref{Prop1} has the following consequence. We set
\[
\overline{J}_{AB}=J_{A}\oplus(-J_{B})=\overline{I}_{AB}J_{AB}\overline{I}_{AB}%
\]
(that is, $\overline{J}_{AB}$ is the standard symplectic matrix of the
symplectic vector space $(\mathbb{R}^{2n_{A}}\oplus\mathbb{R}^{2n_{B}}%
,\sigma_{A}\oplus(-\sigma_{B})$).

\begin{corollary}
\label{Cor1}Let $\widehat{\rho}=(2\pi\hbar)^{n}\operatorname*{Op}%
\nolimits_{\mathrm{W}}(\rho)$ be a separable density operator. Then, in
addition to (\ref{quant2}), we have%
\begin{equation}
\Sigma+\frac{i\hbar}{2}\overline{J}_{AB}\geq0~; \label{quant3}%
\end{equation}
or equivalently%
\begin{equation}
\overline{\Sigma}+\frac{i\hbar}{2}J_{AB}\geq0 \label{quant4}%
\end{equation}
where $\overline{\Sigma}=\overline{I}_{AB}\Sigma\overline{I}_{AB}$ that is%
\begin{equation}
\overline{\Sigma}=%
\begin{pmatrix}
\Sigma_{AA} & \Sigma_{AB}\overline{I}_{B}\\
\overline{I}_{B}\Sigma_{BA} & \overline{I}_{B}\Sigma_{BB}\overline{I}_{B}%
\end{pmatrix}
~. \label{quantbar}%
\end{equation}

\end{corollary}

\begin{proof}
Replacing $\rho$ with $\rho\circ\overline{I}_{AB}$ the matrix $\Sigma_{AA}$ in
(\ref{sigmab}) remains unchanged while $\Sigma_{BB}$, $\Sigma_{AB}$, and
$\Sigma_{BA}$ become $\overline{I}_{B}\Sigma_{BB}\overline{I}_{B}$,
$\Sigma_{AB}\overline{I}_{B}$, and $\overline{I}_{B}\Sigma_{BA}$ respectively.
The covariance matrix (\ref{sigmab}) thus becomes $\overline{\Sigma}%
=\overline{I}_{AB}\Sigma\overline{I}_{AB}$. In view of Proposition \ref{Prop1}
the operator $(2\pi\hbar)^{n}\operatorname*{Op}\nolimits_{\mathrm{W}}%
(\rho\circ\overline{I}_{B})$ is also positive semidefinite hence we must have
$\overline{\Sigma}+\frac{i\hbar}{2}J_{AB}\geq0$ which is equivalent to
$\Sigma+\frac{i\hbar}{2}\overline{I}_{AB}J_{AB}\overline{I}_{AB}\geq0$. Since
$\overline{I}_{AB}J_{AB}\overline{I}_{AB}=\overline{J}_{AB}$ this is
equivalent to (\ref{quant3}).
\end{proof}

The ellipsoid
\begin{equation}
\overline{\Omega}=\{z\in\mathbb{R}^{2n}:\tfrac{1}{2}\overline{\Sigma}%
^{-1}z^{2}\leq1\} \label{covellbar1}%
\end{equation}
for the covariance matrix of the partial transpose $\widehat{\rho}^{T_{B}}$
can be expressed in terms of the matrix $\overline{M}=\frac{\hbar}{2}%
\overline{\Sigma}^{-1}$ by%
\begin{equation}
\overline{\Omega}=\{z\in\mathbb{R}^{2n}:\overline{M}z^{2}\leq\hbar\}
\label{covellbar2}%
\end{equation}
where $\overline{M}=\overline{I}_{AB}M\overline{I}_{AB}$.

\subsection{Werner and Wolf's condition}

Using techniques previously developed by Werner \cite{Werner}, Werner and Wolf
\cite{ww1} prove the following crucial necessary condition for separability (a
different proof can be found in Serafini \cite{Serafini}, p.178):

\begin{proposition}
[Werner and Wolf]\label{propww}Suppose that the density operator
$\widehat{\rho}$ with covariance matrix $\Sigma$ is separable. There exist two
partial covariance matrices $\Sigma_{A}$ and $\Sigma_{B}$ of dimensions
$2n_{A}\times2n_{A}$ and $2n_{B}\times2n_{B}$ satisfying the quantum
conditions
\begin{equation}
\Sigma_{A}+\frac{i\hbar}{2}J_{A}\geq0\text{ \ and \ }\Sigma_{B}+\frac{i\hbar
}{2}J_{B}\geq0 \label{quantAB}%
\end{equation}
and such that%
\begin{equation}
\Sigma\geq\Sigma_{A}\oplus\Sigma_{B}~. \label{ww2}%
\end{equation}

\end{proposition}

We are going to show that Werner and Wolf's result can be considerably refined
using the properties of the symplectic group. We first remark that the quantum
condition $\Sigma+\frac{i\hbar}{2}J\geq0$ on a covariance matrix is equivalent
to the following property: there exists $S\in\operatorname*{Sp}(n)$ such that
$\Sigma\geq\frac{\hbar}{2}(S^{T}S)^{-1}$ (see \cite{FOOP,Birkbis}); this
property is easily deduced from (\ref{quantcamelbis}). It is equivalent to
saying that the covariance ellipsoid $\Omega$ contains a \textit{quantum blob}
\cite{blob}.

\begin{proposition}
\label{Prop2}The Werner--Wolf condition (\ref{ww2}) is equivalent to the
existence of two positive definite symplectic matrices
\begin{equation}
P_{A}=(S_{A}^{T}S_{A})^{-1} \text{ \ },\text{ \ }P_{B}=(S_{B}^{T}S_{B})^{-1}
~, \label{pab}%
\end{equation}
with $S_{A}\in\operatorname*{Sp}(n_{A})$ and $S_{B}\in\operatorname*{Sp}%
(n_{B})$, such that
\begin{equation}
\Sigma\geq\frac{\hbar}{2}(P_{A}\oplus P_{B})~. \label{papb1}%
\end{equation}
Equivalently, the covariance ellipsoid $\Omega$ contains a quantum blob of the
form $(S_{A}\oplus S_{B})(B^{2n}(\sqrt{\hbar}))$.
\end{proposition}

\begin{proof}
The sufficiency of the condition is clear since $\Sigma_{A}=\frac{\hbar}%
{2}P_{A}$ and $\Sigma_{A}=\frac{\hbar}{2}P_{A}$ satisfy the conditions
(\ref{quantAB}). Assume conversely that $\Sigma\geq\Sigma_{A}\oplus\Sigma_{B}$
as in Proposition \ref{propww}. In view of Williamson's diagonalization
theorem \cite{Folland,Birk} there exist $S_{A}\in\operatorname*{Sp}(n_{A})$
and $S_{B}\in\operatorname*{Sp}(n_{B})$ such that $S_{A}\Sigma_{A}S_{A}%
^{T}=D_{A}$ and $S_{B}\Sigma_{B}S_{B}^{T}=D_{B}$ where%
\[
D_{A}=%
\begin{pmatrix}
\Lambda_{A} & 0\\
0 & \Lambda_{A}%
\end{pmatrix}
\text{ \ },\text{ }D_{B}=%
\begin{pmatrix}
\Lambda_{B} & 0\\
0 & \Lambda_{B}%
\end{pmatrix}
\]
and $\Lambda_{A}$, $\Lambda_{B}$ being the diagonal matrices consisting of the
symplectic eigenvalues $\lambda_{1}^{\sigma_{A}},...,\lambda_{n_{A}}%
^{\sigma_{A}}$ of $\Sigma_{A}$ and $\lambda_{1}^{\sigma_{B}},...,\lambda
_{n_{B}}^{\sigma_{B}}$ of $\Sigma_{B}$ (\textit{i.e}. the $\pm i\lambda
_{1}^{\sigma_{A}}$ are the eigenvalues of of $\Sigma_{A}^{1/2}J_{A}\Sigma
_{A}^{1/2}$, see \textit{e.g.} \cite{Birk,sisumu}). Since $S_{A}J_{A}S_{A}%
^{T}=J_{A}$ and $S_{B}J_{B}S_{B}^{T}=J_{B}$ the conditions $\Sigma_{A}%
+\frac{i\hbar}{2}J_{A}\geq0$ and $\Sigma_{B}+\frac{i\hbar}{2}J_{B}\geq0$ are
equivalent to $D_{A}+\frac{i\hbar}{2}J_{A}\geq0$ and $D_{B}+\frac{i\hbar}%
{2}J_{B}\geq0$. These conditions imply that $D_{A}\geq\frac{\hbar}{2}I_{A}$
and $D_{B}\geq\frac{\hbar}{2}I_{B}$: the characteristic equation of
$D_{A}+\frac{i\hbar}{2}J_{A}$ is
\[
\det\left(  (\Lambda_{A}-\lambda I_{A})^{2}-\tfrac{1}{4}\hbar^{2}I_{A}\right)
=0~.
\]
Writing $\Lambda_{A}=\operatorname*{diag}(\lambda_{1}^{\sigma_{A}}%
,...,\lambda_{n}^{\sigma_{A}})$ this equation is equivalent to the set of
equations
\[
(\lambda_{j}^{\sigma_{A}}-\lambda)^{2}-\tfrac{1}{4}\hbar^{2}=0\text{ , }1\leq
j\leq n_{A},
\]
whose solutions are the real numbers $\lambda_{j}=\lambda_{j}^{\sigma_{A}}%
\pm\frac{\hbar}{2}$. Since $\lambda_{j}\geq0$ we must thus have $\lambda
_{j}^{\sigma_{A}}\geq\frac{\hbar}{2}$ and hence $D_{A}\geq\frac{\hbar}{2}%
I_{A}$. Similarly, $D_{B}\geq\frac{\hbar}{2}I_{B}$ so we must have the
inequalities%
\begin{align*}
\Sigma_{A}  &  =S_{A}^{-1}D_{A}(S_{A}^{T})^{-1}\geq\frac{\hbar}{2}(S_{A}%
^{T}S_{A})^{-1}\\
\Sigma_{B}  &  =S_{B}^{-1}D_{A}(S_{B}^{T})^{-1}\geq\frac{\hbar}{2} (S_{B}%
^{T}S_{B})^{-1}~.
\end{align*}
Setting $P_{A}=(S_{A}^{T}S_{A})^{-1}$ and $P_{B}=(S_{B}^{T}S_{B})^{-1}$ the
inequality (\ref{papb1}) follows.
\end{proof}

\subsection{A property of the reduced covariance ellipsoids}

\label{section2.3}

The previous propositions have a very simple geometrical meaning, to which we
will come back in Section \ref{secsuff1}. The conditions (\ref{quantAB}) mean
that $\Sigma_{A}$ and $\Sigma_{B}$ are quantum covariances matrices, hence the
sum
\[
\Sigma_{A}\oplus\Sigma_{B}\equiv%
\begin{pmatrix}
\Sigma_{A} & 0\\
0 & \Sigma_{B}%
\end{pmatrix}
\]
is a quantum covariance matrix in its own right. It follows from (\ref{ww2})
that the corresponding covariance ellipsoid, which we denote
\begin{equation}
\label{AopB}\Omega_{A\oplus B}=\{z_{A}\oplus z_{B}:\tfrac{1}{2}\Sigma_{A}%
^{-1}z_{A} ^{2}+\tfrac{1}{2}\Sigma_{B}^{-1}z_{B}^{2}\leq1\} ~,
\end{equation}
is included in $\Omega$.

Moreover, in view of (the proof of) Proposition \ref{Prop2}, the ellipsoid
$\Omega_{A \oplus B}$ always contains a quantum blob of the form
\begin{equation}
\label{omab}\Omega_{AB}=S_{A}\oplus S_{B}(B^{2n}(\sqrt{\hbar}))=\{ z_{A}\oplus
z_{B}:|S_{A}^{-1}z_{A}|^{2}+|S_{B}^{-1}z_{B}|^{2}\leq\hbar\} ~.
\end{equation}
Hence, if the density operator $\widehat{\rho}$ with covariance ellipsoid
$\Omega$ is separable then there exist quantum covariance ellipsoids of the
form (\ref{AopB}) and (\ref{omab}) such that the following inclusions hold
\begin{equation}
\label{inclusion}\Omega\supset\Omega_{A\oplus B} \supset\Omega_{AB} ~.
\end{equation}

This result has an interesting consequence for the covariance ellipsoids
\[
\Omega_{A}=\{z_{A}:\tfrac{1}{2}\Sigma_{AA}^{-1}z_{A}^{2}\leq1\}\text{ and
}\Omega_{B}=\{z_{B}:\tfrac{1}{2}\Sigma_{BB}^{-1}z_{B}^{2}\leq1\}~.
\]
of the reduced density operators $\widehat{\rho}_{A}$ and $\widehat{\rho}_{B}%
$. We first show that:

\begin{proposition}
\label{PropABBO}The orthogonal projections $\Pi_{A}\Omega_{AB}$ and $\Pi
_{B}\Omega_{AB}$ of $\Omega_{AB}$ onto $\mathbb{R}^{2n_{A}}$ and
$\mathbb{R}^{2n_{B}}$ satisfy
\begin{equation}
\label{piomab}\Pi_{A}\Omega_{AB}=S_{A}(B^{2n_{A}}(\sqrt{\hbar}))\text{
\ },\text{ \ }\Pi_{B}\Omega_{AB}=S_{B}(B^{2n_{B}}(\sqrt{\hbar}))\text{~.}%
\end{equation}

\end{proposition}

\begin{proof}
This result is easily proved directly from the definition of $\Omega_{AB}$.
Alternatively we can use a recent result \cite{digopra19} which generalizes
Gromov's symplectic non-squeezing theorem \cite{Gromov} in the linear case,
and which refines a previous result of Abbondandolo and his collaborators
\cite{abbo,abbobis}. It states that for every $S\in\operatorname*{Sp}(n)$
there exists $S_{A}\in\operatorname*{Sp}(n_{A})$ such that
\begin{equation}
\label{abboprime}\Pi_{A}S(B^{2n}(\sqrt{\hbar}))\supset S_{A}(B^{2n_{A}}%
(\sqrt{\hbar}))
\end{equation}
with equality if and only if $S=S_{A}\oplus S_{B}$. The result (\ref{piomab})
follows using the definition (\ref{omab}) of $\Omega_{AB}$. The same argument
applies to $\Pi_{B}\Omega_{AB}$.
\end{proof}

Notice that since symplectic automorphisms are volume-preserving the result
above implies that
\begin{align*}
\operatorname*{Vol}\nolimits_{2n_{A}}\Pi_{A}\Omega_{AB}  &
=\operatorname*{Vol}\nolimits_{2n_{A}}B^{2n_{A}}(\sqrt{\hbar})=\frac{(\pi
\hbar)^{n_{A}}}{n_{A}!}\\
\operatorname*{Vol}\nolimits_{2n_{B}}\Pi_{B}\Omega_{AB}  &
=\operatorname*{Vol}\nolimits_{2n_{B}}B^{2n_{B}}(\sqrt{\hbar})=\frac{(\pi
\hbar)^{n_{B}}}{n_{B}!}~.
\end{align*}

Likewise, the orthogonal projections of $\Omega_{A\oplus B}$ on $\mathbb{R}%
^{2n_{A}}$ and $\mathbb{R}^{2n_{B}}$ are just the intersections of
$\Omega_{A\oplus B}$ with the hyperplanes $z_{B}=0$ and $z_{A}=0$, respectively.

Finally, from (\ref{inclusion}) we easily conclude that the covariant
ellipsoids $\Omega_{A}$ and $\Omega_{B}$ impose the following constraints on
the symplectic matrices $S_{A}$ and $S_{B}$ of Proposition \ref{Prop2}:

\begin{corollary}
\label{CorABBO} Assume that the density operator $\widehat{\rho}$ with
covariant ellipsoid $\Omega$ is separable. Then the symplectic matrices
$S_{A}$ and $S_{B}$ of Proposition \ref{Prop2} satisfy:
\begin{equation}
\label{Pellipsoids}S_{A}B^{2n_{A}}(\sqrt{\hbar}) \subset\Omega_{A} \quad,
\quad S_{B}B^{2n_{B}}(\sqrt{\hbar}) \subset\Omega_{B}%
\end{equation}

\end{corollary}

\begin{proof}
From (\ref{inclusion}) we have $\Omega_{AB} \subset\Omega$ and so:
\[
\Pi_{A} \Omega_{AB} \subset\Pi_{A} \Omega= \Omega_{A} \quad\mbox{and} \quad
\Pi_{B} \Omega_{AB} \subset\Pi_{B} \Omega=\Omega_{B} ~.
\]
The result then follows from (\ref{piomab}).
\end{proof}

\section{Gaussian Quantum States}

\label{Section3}

\subsection{Generalities, a sufficient condition for separability}

A simple, but very interesting case, occurs when $\rho$ is a Gaussian Wigner
distribution
\[
\rho(z)=\frac{1}{(2\pi)^{n}\sqrt{\det\Sigma}}e^{-\frac{1}{2}\Sigma^{-1}%
(z-\bar{z})^{2}}%
\]
centered at $\bar{z}\in\mathbb{R}^{2n}$, where $\Sigma$ is a positive definite
real symmetric $2n\times2n$ matrix (the \textquotedblleft covariance
matrix\textquotedblright). The normalization factor preceding the exponential
guarantees that $\operatorname*{Tr}(\widehat{\rho})=1$. We will only consider
the case $\bar{z}=0$; the more general case is easily reduced to the former by
a phase space translation. Hence we assume that
\begin{equation}
\rho(z)=\frac{1}{(2\pi)^{n}\sqrt{\det\Sigma}}e^{-\frac{1}{2}\Sigma^{-1}z^{2}}
\label{Gauss}%
\end{equation}
and, setting as usual $M=\frac{\hbar}{2}\Sigma^{-1}$, we can rewrite
(\ref{Gauss}) as
\begin{equation}
\rho(z)=(\pi\hbar)^{-n}(\det M)^{1/2}e^{-\frac{1}{\hbar}Mz^{2}}~. \label{rhum}%
\end{equation}
Since $\rho$ is real, the Weyl operator $\widehat{\rho}=(2\pi\hbar
)^{n}\operatorname*{Op}\nolimits_{\mathrm{W}}(\rho)$ is self-adjoint. To
ensure that $\widehat{\rho}$ is positive semidefinite it is \textit{necessary
and sufficient} \cite{Narcow89,Narcow90,Narconnell} that the covariance matrix
satisfies the quantum condition (\ref{quant2}), which we assume from now on.
Notice that the general result (\ref{abboprime}) that was used in Proposition
\ref{PropABBO} also provides an alternative proof of the fact that the partial
trace operators $\widehat{\rho}_{A}$ and $\widehat{\rho}_{B}$ are density
operators. In fact, to prove this we had to use for the general case the KLM
conditions (Proposition \ref{PropKLM}) in Section \ref{subsec22} to prove the
positivity properties $\widehat{\rho}_{A}\geq0$ and $\widehat{\rho}_{B}\geq0$.
In the Gaussian case we can instead consider the quantum condition
(\ref{quant}) which is equivalent to (\ref{quantcamelbis}). From
(\ref{abboprime}) it then follows that
\begin{equation}
\Sigma_{A}+\frac{i\hbar}{2}J_{A}\geq0\text{ ,}\ \Sigma_{B}+\frac{i\hbar}%
{2}J_{B}\geq0 \label{sigab}%
\end{equation}
hence $\widehat{\rho}_{A}$ and $\widehat{\rho}_{B}$ are (Gaussian) density operators.

The purity of $\widehat{\rho}$ is then given by%
\begin{equation}
\mu(\widehat{\rho})=\left(  \tfrac{\hbar}{2}\right)  ^{n}(\det\Sigma
)^{-1/2}=\sqrt{\det M} \label{purity}%
\end{equation}
(see \textit{e.g.} \cite{Birk}, \S 9.3, p.301). That the terminology
\textquotedblleft covariance matrix\textquotedblright\ applied to $\Sigma$ is
justified in the quantum case as it is in classical statistical mechanics,
follows from formulas (\ref{zalpha}) and (\ref{cov}). It is also clear that we
have $\rho\in\Gamma^{m}(\mathbb{R}^{2n})$ for every $m<-2n$ hence $\rho_{A}%
\in\Gamma^{m_{A}}(\mathbb{R}^{2n_{A}})$ for every $m_{A}<-2n_{A}$ (see
Proposition \ref{mA}).

It turns out that Werner and Wolf's conditions in Proposition \ref{propww} are
sufficient for a Gaussian state to be separable:

\begin{proposition}
\label{propwwg}Assume that there exist two partial covariance matrices
$\Sigma_{A}$ and $\Sigma_{B}$ satisfying the quantum conditions (\ref{sigab})
and such that
\begin{equation}
\Sigma\geq\Sigma_{A}\oplus\Sigma_{B}~. \label{sigsigab}%
\end{equation}
Then the Gaussian state (\ref{Gauss}) is separable.
\end{proposition}

\begin{proof}
See \cite{ww1} (Proposition 1).
\end{proof}

\subsection{Pure Gaussians}

Let $X$ and $Y$ be real symmetric $n\times n$ matrices, with $X>0$. To these
matrices we associate the Gaussian function $\phi_{X,Y}$ on $\mathbb{R}^{n}$
defined by%
\begin{equation}
\phi_{X,Y}(x)=(\pi\hbar)^{-n/4}(\det X)^{1/4}e^{-\frac{1}{2\hbar}(X+iY)x^{2}}
\label{fxy1}%
\end{equation}
where we are writing $(X+iY)x^{2}$ for $(X+iY)x\cdot x$. This function is
$L^{2}$-normalized: $||\phi_{X,Y}||_{L^{2}(\mathbb{R}^{n})}=1$ and its Wigner
transform is given by the well-known formula \cite{Bastiaans,Birk,Wigner}
\begin{equation}
W\phi_{X,Y}(z)=(\pi\hbar)^{-n}e^{-\tfrac{1}{\hbar}Gz^{2}} \label{phagauss}%
\end{equation}
where $G$ is the positive-definite symmetric matrix
\begin{equation}
G=%
\begin{pmatrix}
X+YX^{-1}Y & YX^{-1}\\
X^{-1}Y & X^{-1}%
\end{pmatrix}
~. \label{gsym}%
\end{equation}
In fact $G=S^{T}S$ where
\begin{equation}
S=%
\begin{pmatrix}
X^{1/2} & 0\\
X^{-1/2}Y & X^{-1/2}%
\end{pmatrix}
\in\operatorname*{Sp}(n) \label{bi}%
\end{equation}
hence $G$ is a positive definite symplectic matrix. Setting $\Sigma^{-1}%
=\frac{\hbar}{2}G$ we can rewrite (\ref{phagauss}) as
\[
W\phi_{X,Y}(z)=\frac{1}{(2\pi)^{n}\sqrt{\det\Sigma}}e^{-\frac{1}{2}\Sigma
^{-1}z^{2}}~.
\]
Hence, to $\rho_{X,Y}=W\phi_{X,Y}$ corresponds a Gaussian density operator
$\widehat{\rho}_{X,Y}$ (the quantum condition (\ref{quant}) becomes here
$S^{T}S+iJ\geq0$; since $(S^{T})^{-1}JS^{-1}=J$ this is equivalent to
$I+iJ\geq0$ which is trivially satisfied).

\begin{lemma}
\label{Lemmapure}A Gaussian state $\widehat{\rho}$ is pure if and only if
there exists $(X,Y)$ such that $\rho=W\phi_{X,Y}$.
\end{lemma}

\begin{proof}
\ The sufficiency is clear, so all we have to do is to show that it is
necessary. The purity formula (\ref{purity}) for Gaussians shows that
$\mu(\widehat{\rho})=1$ if and only if $\det\Sigma=(\hbar/2)^{2n}$. Let
$\lambda_{1}^{\sigma},...,\lambda_{n}^{\sigma}$ be the symplectic eigenvalues
of $\Sigma$ (\textit{i.e.} the numbers $\lambda_{j}^{\sigma}>0$ such that the
$\pm i\lambda_{j}^{\sigma}$ are the eigenvalues of $JM$); in view of
Williamson's symplectic diagonalization theorem there exists $S\in
\operatorname*{Sp}(n)$ such that $\Sigma=(S^{T})^{-1}DS^{-1}$ where $D=%
\begin{pmatrix}
\Lambda & 0\\
0 & \Lambda
\end{pmatrix}
$ with $\Lambda=\operatorname*{diag}(\lambda_{1}^{\sigma},...,\lambda
_{n}^{\sigma})$. The quantum condition (\ref{quant}) is equivalent to
$\lambda_{j}^{\sigma}\geq\hbar/2$ for all $j$ hence
\[
\det\Sigma=(\lambda_{1}^{\sigma})^{2}\cdot\cdot\cdot(\lambda_{n}^{\sigma}%
)^{2}=1
\]
if and only all the $\lambda_{j}^{\sigma}$ are equal to $\hbar/2$, hence
$\Sigma=\frac{\hbar}{2}(S^{T})^{-1}S^{-1}$.
\end{proof}

\begin{remark}
The action of the metaplectic group $\operatorname*{Mp}(n)$ on the set of all
Gaussians $\phi_{X,Y}$ is transitive \cite{Birkbis,Wigner}. The Lemma above
can thus be rephrased by saying that every pure Gaussian state is obtained
from the standard Gaussian $\phi_{0}(x)=(\pi\hbar)^{-n/4}e^{-|x|^{2}/2\hbar}$
by some $\widehat{S}\in\operatorname*{Mp}(n)$.
\end{remark}

\subsection{Separability of Gaussian states}

Before we state and prove our main results, let us make the following simple observation:

\begin{lemma}
\label{LemmaX} If the covariance ellipsoid
\[
\Omega=\{z\in\mathbb{R}^{2n}:\tfrac{1}{2}\Sigma^{-1}z^{2}\leq1\}
\]
of a Gaussian state $\widehat{\rho}$ contains the ball $B^{2n}(\sqrt{\hbar})$,
then $\widehat{\rho}$ is separable for all partitions $(A,B)$.
\end{lemma}

\begin{proof}
Setting $M=\frac{\hbar}{2}\Sigma^{-1}$, the inclusion $B^{2n}(\sqrt{\hbar
})\subset\Omega$ is equivalent to $M\leq I$. Hence, the Werner--Wolf condition
condition (\ref{ww2}) is satisfied with $\Sigma_{A}\oplus\Sigma_{B}%
=\frac{\hslash}{2}I_{2n\times2n}$.
\end{proof}

More generally, there always exists $S\in\operatorname*{Sp}(n)$ such that
$SB^{2n}(\sqrt{\hbar})\subset\Omega$ (see condition (\ref{quantcamelbis})),
but this does not ensure separability unless $S=S_{A}\oplus S_{B}$ with
$S_{A}\in\operatorname*{Sp}(n_{A})$ and $S_{B}\in\operatorname*{Sp}(n_{B})$.
In this case we will have $M\leq S_{A}\oplus S_{B}$ hence (\ref{ww2}) is satisfied.

Next, we are going to show that for Gaussian states the necessary condition
for separability in Proposition \ref{Prop2} is also sufficient.

\begin{proposition}
\label{Prop3}The Gaussian density operator $\widehat{\rho}$ is separable if
and only if there exist positive definite symplectic matrices $P_{A}%
\in\operatorname*{Sp}(n_{A})$ and $P_{B}\in\operatorname*{Sp}(n_{B})$ such
that
\begin{equation}
\Sigma\geq\frac{\hbar}{2}(P_{A}\oplus P_{B})~. \label{MC}%
\end{equation}

\end{proposition}

\begin{proof}
In view of Proposition \ref{Prop2}, the condition (\ref{MC}) is equivalent to
the Werner-Wolf condition (\ref{ww2}). Since for Gaussians the Werner-Wolf
condition is necessary and sufficient, this is also the case for the condition
(\ref{MC}).
\end{proof}

Suppose we have equality in (\ref{MC}). Then $\widehat{\rho}$ is a tensor
product $\widehat{S}_{A}^{-1}\phi_{0,A}\otimes\widehat{S}_{B}^{-1}\phi_{0,B}$
where
\begin{align*}
\phi_{0,A}(x_{A})  &  =(\pi\hbar)^{-n_{A}/4}e^{-|x_{A}|^{2}/2\hbar}\\
\phi_{0,B}(x_{B})  &  =(\pi\hbar)^{-n_{B}/4}e^{-|x_{B}|^{2}/2\hbar}%
\end{align*}
are the standard Gaussians on $\mathbb{R}^{n_{A}}$ and $\mathbb{R}^{n_{B}}$,
and $\widehat{S}_{A}\in\operatorname*{Mp}(n_{A})$ (\textit{resp}.
$\widehat{S}_{B}\in\operatorname*{Mp}(n_{B})$) is anyone of the two
metaplectic operators covering\ $S_{A}$ (\textit{resp.} $S_{B}$). In fact, the
Wigner distribution $\rho$ becomes in this case
\begin{align*}
\rho(z)  &  =(\pi\hbar)^{-n}e^{-\frac{1}{\hbar}(S_{A}^{T}S_{A}z_{A}\cdot
z_{A}+S_{B}^{T}S_{B}z_{B}\cdot z_{B})}\\
&  =W_{A}\phi_{0,A}(S_{A}z_{A})W_{B}\phi_{0,B}(S_{B}z_{B})
\end{align*}
where $W_{A}\phi_{0,A}$ is the Wigner transform of $\phi_{0,A}$ and $W_{B}%
\phi_{0,B}$ that of $\phi_{0,B}$. It follows from the symplectic covariance
property \cite{Wigner} of the Wigner transform that
\[
W_{A}\phi_{0,A}\circ S_{A}=W_{A}(\widehat{S}_{A}^{-1}\phi_{0,A})\text{ \ ,
\ }W_{B}\phi_{0,B}\circ S_{B}=W_{A}(\widehat{S}_{B}^{-1}\phi_{0,B})
\]
hence $\rho$ is the Wigner transform of $\widehat{S}_{A}^{-1}\phi_{0,A}%
\otimes\widehat{S}_{B}^{-1}\phi_{0,B}$. The converse of this property is
trivial. Notice that $\widehat{S}_{A}^{-1}\phi_{0,A}$ and $\widehat{S}%
_{B}^{-1}\phi_{0,B}$ are easily calculated \cite{Birk,Wigner}: they are
explicitly given by
\begin{align*}
\widehat{S}_{A}^{-1}\phi_{0,A}(x_{A})  &  =(\pi\hbar)^{-n_{A}/4}(\det
X_{A})^{1/4}e^{-\frac{1}{2\hbar}(X_{A}+iY_{A})x_{A}\cdot x_{A}}\\
\widehat{S}_{B}^{-1}\phi_{0,B}(x_{B})  &  =(\pi\hbar)^{-n_{B}/4}(\det
X_{B})^{1/4}e^{-\frac{1}{2\hbar}(X_{B}+iY_{B})x_{B}\cdot x_{B}}%
\end{align*}
where the real symmetric matrices $X_{A}>0$, $X_{B}>0$ and $Y_{A},Y_{B}$ are
obtained by solving the identities%
\begin{align*}
S_{A}^{T}S_{A}  &  =%
\begin{pmatrix}
X_{A}+Y_{A}X_{A}^{-1}Y_{A} & Y_{A}X_{A}^{-1}\\
X_{A}^{-1}Y_{A} & X_{A}^{-1}%
\end{pmatrix}
\\
S_{B}^{T}S_{B}  &  =%
\begin{pmatrix}
X_{B}+Y_{B}X_{B}^{-1}Y_{B} & Y_{B}X_{B}^{-1}\\
X_{B}^{-1}Y_{B} & X_{B}^{-1}%
\end{pmatrix}
~.
\end{align*}

More generally the Gaussian state $\widehat{\rho}$ is separable if and only if
its Wigner distribution dominates a tensor product of two Gaussian states, up
to a factor being the purity of $\widehat{\rho}$:

\begin{theorem}
\label{ThmAB}The Gaussian state $\widehat{\rho}$ is separable if and only if
there exist pairs $(X_{A},Y_{A})$ and $(X_{B},Y_{B})$ such that
\begin{equation}
\rho\geq\mu(\widehat{\rho})(W_{A}\phi_{X_{A},Y_{A}}\otimes W_{B}\phi
_{X_{B},Y_{B}}) \label{sepurity}%
\end{equation}
where
\[
\mu(\widehat{\rho})=\left(  \frac{\hbar}{2}\right)  ^{n}(\det\Sigma)^{-1/2}%
\]
is the purity (\ref{purity}) of $\widehat{\rho}$.
\end{theorem}

\begin{proof}
In view of the transitivity of the action of the metaplectic group on
Gaussians, this is equivalent to proving that there exist $\widehat{S_{A}}%
\in\operatorname*{Mp}(n_{A})$ and $\widehat{S_{B}}\in\operatorname*{Mp}%
(n_{B})$ such that
\begin{equation}
\rho\geq\mu(\widehat{\rho})\left(  W_{A}(\widehat{S_{A}}^{-1}\phi
_{0,A})\otimes W_{B}(\widehat{S_{B}}^{-1}\phi_{0,B})\right)  ~. \label{corab2}%
\end{equation}
In view of Proposition \ref{Prop2} $\widehat{\rho}$ is separable if and only
if condition (\ref{MC})
\[
\Sigma\geq\frac{\hbar}{2}\left[  (S_{A}^{T}S_{A})^{-1}\oplus(S_{B}^{T}%
S_{B})^{-1}\right]
\]
holds for some $S_{A}\in\operatorname*{Sp}(n_{A})$ and $S_{B}\in
\operatorname*{Sp}(n_{B})$. Suppose this is the case; by definition
(\ref{Gauss}) of $\rho$ we then have
\[
\rho(z)\geq\frac{1}{(2\pi)^{n}\sqrt{\det\Sigma}}e^{-\frac{1}{\hbar}S_{A}%
^{T}S_{A}z_{A}\cdot z_{A}}e^{-\frac{1}{\hbar}S_{B}^{T}S_{B}z_{B}\cdot z_{B}%
}~.
\]
We have \cite{Birk,Wigner}%
\begin{align*}
W_{A}\phi_{0,A}(S_{A}z_{A})  &  =(\pi\hbar)^{-n_{A}}e^{-\frac{1}{\hbar}|S_{A}
z_{A}|^{2}}\\
W_{B}\phi_{0,B}(S_{B}z_{B})  &  =(\pi\hbar)^{-n_{B}}e^{-\frac{1}{\hbar}|S_{B}
z_{B}|^{2}}%
\end{align*}
and hence%
\begin{equation}
\rho(z)\geq\left(  \frac{\hbar}{2}\right)  ^{n}(\det\Sigma)^{-1/2}W_{A}%
\phi_{0,A}(S_{A}z_{A})W_{B}\phi_{0,B}(S_{B}z_{B})~. \label{corab3}%
\end{equation}
Let now $\widehat{S_{A}}\in\operatorname*{Mp}(n_{A})$ (\textit{resp}.
$\widehat{S_{B}}\in\operatorname*{Mp}(n_{B})$) cover $S_{A}$ (\textit{resp}.
$S_{B}$); we have, using the symplectic covariance of the Wigner transform
\cite{Folland,Birkbis,Wigner}
\begin{align*}
W_{A}\phi_{0,A}(S_{A}z_{A})  &  =W_{A}(\widehat{S_{A}}^{-1} \phi)(z_{A})\\
W_{B}\phi_{0,B}(S_{A}z_{B})  &  =W_{B}(\widehat{S_{B}}^{-1}\phi)(z_{B})
\end{align*}
which shows that (\ref{corab2}) must hold if the state $\widehat{\rho}$ is
separable. Suppose conversely that this inequality holds. Then we must have
\[
e^{-\frac{1}{2}\Sigma^{-1}z\cdot z}\geq e^{-\frac{1}{\hbar}S_{A}^{T}S_{A}%
z_{A}\cdot z_{A}}e^{-\frac{1}{\hbar}S_{B}^{T}S_{B}z_{B}\cdot z_{B})}%
\]
which is equivalent to condition (\ref{MC}) in Proposition \ref{Prop3}.
\end{proof}

\begin{corollary}
If the Gaussian state $\widehat{\rho}$ is separable there exist Gaussians
$\phi_{X_{A},Y_{A}}$ and $\phi_{X_{B},Y_{B}}$ such that%
\begin{equation}
\rho_{A}\geq\mu(\widehat{\rho})W_{A}\phi_{X_{A},Y_{A}}\text{ \ , \ }\rho
_{B}\geq\mu(\widehat{\rho})W_{B}\phi_{X_{B},Y_{B}}~. \label{subGauss}%
\end{equation}

\end{corollary}

\begin{proof}
It immediately follows from the inequality (\ref{sepurity}) integrating $\rho$
with respect to $z_{B}$ and $z_{A}$.
\end{proof}

Let us describe in detail the reduced states of a Gaussian state:

\begin{proposition}
\label{Thm1}The reduced density operator $\widehat{\rho}_{A}$ is a Gaussian
state with Wigner distribution
\begin{equation}
\rho_{A}(z_{A})=(\pi\hbar)^{-n_{A}}(\det M/M_{BB})^{1/2}e^{-\frac{1}{\hbar
}(M/M_{BB})z_{A}^{2}}; \label{rhoaza}%
\end{equation}
and its covariance ellipsoid
\begin{equation}
\Omega_{A}=\{z_{A}:(M/M_{BB})z_{A}^{2}\leq\hbar\} \label{covreda}%
\end{equation}
is the orthogonal projection $\Pi_{A}\Omega$ on $\mathbb{R}^{2n_{A}}$ of the
covariance ellipsoid $\Omega$ of $\widehat{\rho}$.
\end{proposition}

\begin{proof}
The result is in a sense rather obvious since the calculation of $\rho_{A}$
involves the integration of the Gaussian $\rho$ with respect to a partial set
of variables, and thus yields a Gaussian. That this Gaussian is given by
(\ref{rhoaza}) then follows from the projection formula (\ref{pambb}). Let us
however give a direct analytical proof. Writing $z=z_{A}\oplus z_{B}$\ we
have
\[
Mz^{2}=M_{AA}z_{A}^{2}+2M_{BA}z_{A}\cdot z_{B}+M_{BB}z_{B}^{2}%
\]
so that%
\[
\int_{\mathbb{R}^{2n_{B}}}e^{-\frac{1}{\hbar}Mz^{2}}dz_{B}=e^{-\frac{1}{\hbar
}M_{AA}z_{A}^{2}}\int_{\mathbb{R}^{2n_{B}}}e^{-\frac{1}{\hbar}(M_{BB}z_{B}%
^{2}+2M_{BA}z_{A}\cdot z_{B})}dz_{B}~.
\]
Setting $z_{B}=u_{B}-M_{BB}^{-1}M_{BA}z_{A}$ we have
\[
M_{BB}z_{B}^{2}+2M_{BA}z_{A}\cdot z_{B}=M_{BB}u_{B}^{2}-M_{AB}M_{BB}%
^{-1}M_{BA}z_{A}^{2}%
\]
and hence, integrating with respect to the variables $z_{B}$,
\[
\int_{\mathbb{R}^{2n_{B}}}e^{-\frac{1}{\hbar}Mz^{2}}dz_{B}=e^{-\frac{1}{\hbar
}(M_{AA}-M_{AB}M_{BB}^{-1}M_{BA})z_{A}^{2}}\int_{\mathbb{R}^{2n_{B}}}%
e^{-\frac{1}{\hbar}M_{BB}u_{B}^{2}}du_{B}~.
\]
Using the classical formula (Folland \cite{Folland}, App. A)
\[
\int_{\mathbb{R}^{2n_{B}}}e^{-\frac{1}{\hbar}M_{BB}u_{B}^{2}}du_{B}=(\pi
\hbar)^{n_{B}}(\det M_{BB})^{-1/2}%
\]
we thus have
\[
\int_{\mathbb{R}^{2n_{B}}}e^{-\frac{1}{\hbar}Mz^{2}}dz_{B}=(\pi\hbar)^{n_{B}%
}(\det M_{BB})^{-1/2}e^{-\frac{1}{\hbar}(M/M_{BB})z_{A}^{2}}%
\]
where $M/M_{BB}$ is the Schur complement (\ref{Schur}) of $M_{BB}$ of $M$; the
identity (\ref{rhoaza}) now follows from formula (\ref{schurdet}). The
covariance ellipsoid of the reduced state $\widehat{\rho}_{A}$ is given by
(\ref{covreda}), and in view of Lemma \ref{lemmapim} it is indeed the
orthogonal projection $\Pi_{A}\Omega$ of $\Omega$ on $\mathbb{R}^{2n_{A}}$.
\end{proof}

\begin{corollary}
The purity of the reduced density operator $\widehat{\rho}_{A}$ is
\begin{equation}
\mu(\widehat{\rho}_{A})=(\det M/M_{BB})^{1/2} \label{mura}%
\end{equation}
and $\widehat{\rho}_{A}$ is a pure state if and only if $M/M_{BB}%
\in\operatorname*{Sp}(n_{A})$, in which case case we have $\mu(\widehat{\rho
})=\det M_{BB}$.
\end{corollary}

\begin{proof}
The purity of $\widehat{\rho}_{A}$ is $\mu(\widehat{\rho}_{A})=\sqrt{\det
M/M_{BB}}$; hence $\mu(\widehat{\rho}_{A})=1$ if and only if $\det M/M_{BB}%
=1$; by the same token as used in Lemma \ref{Lemmapure} we must then have
$M/M_{BB}\in\operatorname*{Sp}(n_{A})$. The equality $\mu(\widehat{\rho})=\det
M_{BB}$ follows from the identity (\ref{schurdet}).
\end{proof}

\section{Sufficient Conditions for Separability of Gaussian states}

\label{secsuff1}

In this section, we will derive a number of sufficient, albeit not necessary,
conditions for the separability of Gaussian states.

We will write as usual
\begin{equation}
M=\frac{\hbar}{2}\Sigma^{-1}=%
\begin{pmatrix}
M_{AA} & M_{AB}\\
M_{BA} & M_{BB}%
\end{pmatrix}
~,
\end{equation}
and it is presupposed that $M=M^{T}>0$, and hence $M_{AA}>0$, $M_{BB}>0$ and
$M_{BA}=M_{AB}^{T}$. It follows from Proposition \ref{propWill} that:%
\begin{gather}
\Sigma+\frac{i\hbar}{2}J_{AB}\geq0\text{ }\Longleftrightarrow\text{\emph{ The
symplectic }}\label{eigenM}\\
\text{\emph{eigenvalues} }\lambda_{\sigma,j} (M) \text{ \emph{of} }M\text{
\emph{are all} }\leq1~.\nonumber
\end{gather}

We shall also assume, without loss of generality, that $n_{B} \geq n_{A}$.
Let
\begin{equation}
\mu_{1}^{AB} \geq\mu_{2}^{AB} \geq\cdots\geq\mu_{2n_{A}}^{AB} \geq0
\label{eqExtra1}%
\end{equation}
be the singular values of $M_{AB}$, that is the positive square roots of the
eigenvalues of the $2n_{A} \times2n_{A}$ matrix $M_{AB}M_{AB}^{T}=M_{AB}%
M_{BA}$. Notice that, apart from the multiplicities of zero, the matrices
$M_{AB}M_{BA}$ and $M_{BA}M_{AB}$ have the same eigenvalues, and so $M_{AB}$
and $M_{BA}$ have the same singular values.

We shall write, as customary, $|M_{AB}|=\left(  M_{AB}M_{BA}\right)  ^{1/2}$
and $|M_{BA}|=\left(  M_{BA}M_{AB}\right)  ^{1/2}$. In particular, we have:
\begin{equation}
\Vert M_{AB}\Vert_{op}=\sup_{z_{B}\neq0}\frac{|M_{AB}z_{B}|}{|z_{B}|}=\mu
_{1}^{AB}=\sup_{z_{A}\neq0}\frac{|M_{BA}z_{A}|}{|z_{A}|}=\Vert M_{BA}%
\Vert_{op}~. \label{eqExtra2}%
\end{equation}

By the singular value decomposition, there exist unitary matrices $U
\in\mathbb{C}^{2n_{A} \times2n_{A}}$ and $V \in\mathbb{C}^{2n_{B} \times
2n_{B}}$, such that
\begin{equation}
M_{AB}=UD_{AB}V^{\ast}~, \label{eqExtra3}%
\end{equation}
where $D_{AB} \in\mathbb{C}^{2n_{A} \times2n_{B}}$ is the diagonal matrix of
singular values, that is $\left(  D_{AB}\right)  _{jj}=\mu_{j}^{AB}$, for
$j=1, \cdots, 2n_{A}$, and $\left(  D_{AB}\right)  _{jk}=0$, for all $j=1,
\cdots, 2n_{A}$ and $k=1, \cdots, 2n_{B}$, such that $j\neq k$.

Given a set of positive numbers $\epsilon=\left(  \epsilon_{1},\cdots
,\epsilon_{2n_{A}}\right)  \in\mathbb{R}_{+}^{2n_{A}}$, we define the
$2n_{A}\times2n_{A}$ matrix $|M_{AB}^{\epsilon}|$ and the $2n_{B}\times2n_{B}$
matrix $|M_{BA}^{\frac{1}{\epsilon}}|$ by:
\begin{equation}%
\begin{array}
[c]{l}%
U^{\ast}|M_{AB}^{\epsilon}|U=\operatorname*{diag}\left(  \epsilon_{1}\mu
_{1}^{AB},\cdots,\epsilon_{2n_{A}}\mu_{2n_{A}}^{AB}\right)  ~,\\
\\
V^{\ast}|M_{BA}^{\frac{1}{\epsilon}}|V=\operatorname*{diag}\left(  \frac
{\mu_{1}^{AB}}{\epsilon_{1}},\cdots,\frac{\mu_{2n_{A}}^{AB}}{\epsilon_{2n_{A}%
}},0,\cdots,0\right)  ~.
\end{array}
\label{eqExtra4}%
\end{equation}
In particular, if we write $1=(1,\cdots,1)$ for $\epsilon_{j}=1$, for all
$j=1,\cdots,2n_{A}$, then we have:
\begin{equation}
|M_{AB}^{1}|=|M_{AB}|\text{ and }|M_{BA}^{1}|=|M_{BA}|~. \label{eqExtra5}%
\end{equation}

We will now derive a hierarchy of sufficient conditions for separability,
which culminate in Theorem \ref{TheoremSep3}. The advantage of developing this
hierarchy, instead of going directly to Theorem \ref{TheoremSep3}, is that in
this manner we increase the computational complexity gradually.

\subsection{The first separability criterion\label{secsepsuff}}

Let us state the first criterion for separability of Gaussian states.

\begin{theorem}
\label{TheoremSep1} Let $\widetilde{M}_{AA}=M_{AA}+\Vert M_{AB}\Vert
_{op}I_{n_{A}}$ and $\widetilde{M}_{BB}=M_{BB}+\Vert M_{BA}\Vert_{op}I_{n_{B}%
}$. If
\begin{equation}
\lambda_{\sigma_{A},j}\left(  \widetilde{M}_{AA}\right)  \leq1\text{ and
}\lambda_{\sigma_{B},k}\left(  \widetilde{M}_{BB}\right)  \leq1~,
\label{eqSep1}%
\end{equation}
for all $j=1,\cdots,n_{A}$ and all $k=1,\cdots,n_{B}$, then the Gaussian state
$\widehat{\rho}$ with covariance ellipsoid
\begin{equation}
\Omega=\left\{  z\in\mathbb{R}^{2n}:Mz^{2}\leq\hbar\right\}  \label{eqSep2}%
\end{equation}
is separable.
\end{theorem}

\begin{proof}
We have, by the Cauchy-Schwarz and the geometric-arithmetic mean
inequalities,
\begin{equation}%
\begin{array}
[c]{c}%
z_{A} \cdot M_{AB} z_{B} \leq|z_{A} \cdot M_{AB} z_{B} | \leq|z_{A}|
\cdot|M_{AB}z_{B}|\\
\\
\leq\|M_{AB}\|_{op} |z_{A}| ~|z_{B}| \leq\frac{\|M_{AB}\|_{op}}{2} \left(
|z_{A}|^{2}+|z_{B}|^{2} \right)  ~.
\end{array}
\label{eqSep3}%
\end{equation}
It follows that
\begin{equation}%
\begin{array}
[c]{c}%
Mz^{2}=M_{AA}z_{A}^{2}+ 2 z_{A} \cdot M_{AB} z_{B} + M_{BB}z_{B}^{2}\\
\\
\leq M_{AA}z_{A}^{2}+ \|M_{AB}\|_{op}|z_{A}|^{2}+ \|M_{AB}\|_{op}|z_{B}|^{2} +
M_{BB}z_{B}^{2} =\\
\\
=\left(  M_{AA} + \| M_{AB}\|_{op} I_{n_{A}}\right)  z_{A}^{2}+ \left(  M_{BB}
+ \| M_{BA}\|_{op} I_{n_{B}}\right)  z_{B}^{2}~,
\end{array}
\label{eqSep4}%
\end{equation}
and thus:
\begin{equation}
M \leq\widetilde{M}_{AA} \oplus\widetilde{M}_{BB}~. \label{eqSep5}%
\end{equation}
If conditions (\ref{eqSep1}) hold, then $\widetilde{M}_{AA}^{-1}+ i J_{A}
\geq0$ and $\widetilde{M}_{BB}^{-1}+ i J_{B} \geq0$. By the Werner-Wolf
condition, the state $\widehat{\rho}$ is separable.
\end{proof}

\subsection{Geometric interpretation}

Here is a straightforward geometric interpretation of Theorem
\ref{TheoremSep1}. It says that if the ellipsoid
\[
\widetilde{\Omega}=\{z\in\mathbb{R}^{2n}:\widetilde{M}_{AA}z_{A}%
^{2}+\widetilde{M}_{BB}z_{B}^{2}\leq\hbar\}
\]
is \textquotedblleft large enough\textquotedblright\ to contain a
\textquotedblleft quantum blob\textquotedblright\ of the type $\Omega
_{AB}=(S_{A}\oplus S_{B})B^{2n}(\sqrt{\hbar})$, then the Gaussian state
$\widehat{\rho}$ with covariance ellipsoid $\Omega$ will be separable. Hence,
we have the inclusions:
\[
\Omega_{AB} \subset\widetilde{\Omega}\subset\Omega
\]
and it follows from the projection results discussed in the sections
\ref{secshadow} and \ref{section2.3} that the following inclusions also hold
\begin{equation}
S_{A}(B^{2n_{A}}(\sqrt{\hbar}))\subset\widetilde{\Omega}_{A} \subset\Omega_{A}
\text{ \ and \ } S_{B}(B^{2n_{B}}(\sqrt{\hbar})) \subset\widetilde{\Omega}_{B}
\subset\Omega_{B}~. \label{projaa}%
\end{equation}
where $\Omega_{A}$ and $\Omega_{B}$ are the covariance ellipsoids of the
reduced density operators $\widehat{\rho}_{A}$ and $\widehat{\rho}_{B}$ (cf.
(\ref{pamaa},\ref{pambb})) and
\begin{equation}
\widetilde{\Omega}_{A}=\{z\in\mathbb{R}^{2n}:\widetilde{M}_{AA}z_{A}^{2}%
\leq\hbar\} \label{projabdef}%
\end{equation}
and likewise for $\widetilde{\Omega}_{B}$.

\subsection{The second separability criterion}

We will now derive a second criterion and then use it to show that the
previous criterion is not necessary for separability of a Gaussian state.

\begin{theorem}
\label{TheoremSep2} Define $M_{AA}^{\sharp}=M_{AA}+ |M_{AB}|$ and
$M_{BB}^{\sharp}=M_{BB}+ |M_{BA}|$. If their symplectic eigenvalues satisfy
\begin{equation}
\lambda_{\sigma_{A},j} \left(  M_{AA}^{\sharp} \right)  \leq1 \text{ and }
\lambda_{\sigma_{B},k} \left(  M_{BB}^{\sharp} \right)  ~, \label{eqSep6}%
\end{equation}
for all $j=1, \cdots, n_{A}$ and all $k=1, \cdots, n_{B}$, then the Gaussian
state $\widehat{\rho}$ with covariance ellipsoid (\ref{eqSep2}) is separable.
\end{theorem}

\begin{proof}
With the previous notation, let $u_{A}=U^{\ast} z_{A}$ and $v_{B}=V^{\ast}
z_{B}$. Then:
\begin{equation}%
\begin{array}
[c]{c}%
z_{A} \cdot M_{AB} z_{B} \leq|z_{A} \cdot M_{AB} z_{B}| =|u_{A} \cdot D_{AB}
v_{B}|= \left|  \sum_{j=1}^{2n_{A}} \mu_{j}^{AB} u_{A,j} v_{B,j} \right| \\
\\
\leq\sum_{j=1}^{2n_{A}} \mu_{j}^{AB}| u_{A,j}|~| v_{B,j}| \leq\sum
_{j=1}^{2n_{A}} \mu_{j}^{AB}\left(  \frac{|u_{A,j}|^{2}}{2} + \frac{
|v_{B,j}|^{2}}{2} \right)  =\\
\\
=\frac{1}{2} \sum_{j=1}^{2n_{A}} \overline{u_{A,j}} \mu_{j}^{AB} u_{A,j} +
\frac{1}{2} \sum_{j=1}^{2n_{A}} \overline{v_{B,j}} \mu_{j}^{AB} v_{B,j}=
\frac{1}{2} |M_{AB}| z_{A}^{2} + \frac{1}{2}|M_{BA}| z_{B}^{2} ~,
\end{array}
\label{eqSep8}%
\end{equation}
where we used (\ref{eqExtra4}) and (\ref{eqExtra5}).

Consequently:
\begin{equation}%
\begin{array}
[c]{c}%
Mz^{2}= M_{AA}z_{A}^{2} + 2 z_{A} \cdot M_{AB} z_{B} + M_{BB} z_{B}^{2}\\
\\
\leq M_{AA}z_{A}^{2} + |M_{AB}| z_{A}^{2} + |M_{BA}| z_{B}^{2} + M_{BB}
z_{B}^{2} = \left(  M_{AA}^{\sharp} \oplus M_{BB}^{\sharp} \right)  z^{2} ~ ,
\end{array}
\label{eqSep9}%
\end{equation}
and the rest follows as before.
\end{proof}

\subsection{An example of non-necessity\label{ExampleImp1}}

Let us now show that the separability criterion stated in Theorem
\ref{TheoremSep1} is sufficient but not necessary. We consider the particular
case $n_{A}=n_{B}=1$.

Let $M$ be the $4\times4$ matrix given by:
\begin{equation}
M=\left(
\begin{array}
[c]{cccc}%
\frac{1}{2} & 0 & \frac{2}{3} & 0\\
0 & \frac{1}{2} & 0 & \frac{1}{4}\\
\frac{2}{3} & 0 & \frac{1}{3} & 0\\
0 & \frac{1}{4} & 0 & \frac{1}{4}%
\end{array}
\right)  ~. \label{eqImp18}%
\end{equation}
With the previous notation, we have;
\begin{equation}%
\begin{array}
[c]{l}%
M_{AA}=\left(
\begin{array}
[c]{cc}%
\frac{1}{2} & 0\\
0 & \frac{1}{2}%
\end{array}
\right)  \text{ \ , \ }M_{BB}=\left(
\begin{array}
[c]{cc}%
\frac{1}{3} & 0\\
0 & \frac{1}{4}%
\end{array}
\right) \\
\\
M_{AB}=M_{BA}=\left\vert M_{AB}\right\vert =\left\vert M_{BA}\right\vert
=\left(
\begin{array}
[c]{cc}%
\frac{2}{3} & 0\\
0 & \frac{1}{4}%
\end{array}
\right)
\end{array}
\label{eqImp19}%
\end{equation}
Since $\mu_{1}^{AB}=\Vert M_{AB}\Vert_{op}=\Vert M_{BA}\Vert_{op}=\frac{2}{3}%
$, we have:
\begin{equation}
\widetilde{M}_{AA}=\left(
\begin{array}
[c]{cc}%
\frac{7}{6} & 0\\
0 & \frac{7}{6}%
\end{array}
\right)  ,\hspace{0.5cm}\widetilde{M}_{BB}=\left(
\begin{array}
[c]{cc}%
1 & 0\\
0 & \frac{11}{12}%
\end{array}
\right)  . \label{eqImp20}%
\end{equation}
It follows that $\lambda_{\sigma_{A}}(\widetilde{M}_{AA})=\frac{7}{6}>1$,
while $\lambda_{\sigma_{B}}(\widetilde{M}_{BB})=\sqrt{\frac{11}{12}}<1$. We
conclude that $M$ does not satisfy the criterion of Theorem \ref{TheoremSep1}.
Nevertheless, $M$ is associated with a separable state. This can be shown
using the criterion of Theorem \ref{TheoremSep2}. Indeed, we have:
\begin{equation}
M_{AA}^{\sharp}=\left(
\begin{array}
[c]{cc}%
\frac{7}{6} & 0\\
0 & \frac{3}{4}%
\end{array}
\right)  ,\hspace{0.5cm}M_{BB}^{\sharp}=\left(
\begin{array}
[c]{cc}%
1 & 0\\
0 & \frac{1}{2}%
\end{array}
\right)  \label{eqImp21}%
\end{equation}
and hence%
\[
\lambda_{\sigma_{A}}(M_{AA}^{\sharp})=\sqrt{\frac{7}{8}}<1\text{ \ and
}\lambda_{\sigma_{B}}(M_{BB}^{\sharp})=\frac{1}{\sqrt{2}}<1~.
\]
According to Theorem \ref{TheoremSep2} the associated Gaussian state is separable.

\subsection{The third separability criterion}

In the previous criteria, we always used the geometric-arithmetic mean
inequality $|ab|\leq(|a|^{2}+|b|^{2})/2$ to prove our results. This inequality
places an upper bound on the product $|ab|$ with $|a|^{2}$ and $|b|^{2}$ on
equal footing. However, it is perfectly conceivable that in some directions
$M_{AA}$ is "too large" for us to have $M_{AA}+\left\vert M_{AB}\right\vert $
dominated by a positive symplectic matrix $P_{A}$ and that this may be
compensated by the fact that $M_{BB}$ is "smaller". In this case, it may be
more suitable to use the scaled geometric-arithmetic mean inequality:
\begin{equation}
|ab|\leq\frac{|a|^{2}}{2\varepsilon}+\frac{\varepsilon|b|^{2}}{2}%
~,\label{eqImp22}%
\end{equation}
which holds for any $\varepsilon>0$. We will derive, using this inequality,
another sufficient criterion for separability, which will permit us to prove
that the criterion stated in Theorem \ref{TheoremSep2} is again sufficient but
not necessary for separability. With the same notation as previously, we have:

\begin{theorem}
\label{TheoremSep3} Let $\widetilde{M}_{AA}^{\epsilon}$ be a $2 n_{A}
\times2n_{A}$ matrix and $\widetilde{M}_{BB}^{\frac{1}{\epsilon}}$ a $2 n_{B}
\times2n_{B}$ defined by:
\begin{equation}
\widetilde{M}_{AA}^{\epsilon}= M_{AA}+ \left\vert M_{AB}^{\epsilon}
\right\vert \text{ and }\widetilde{M}_{BB}^{\frac{1}{\epsilon}}= M_{BB}+
\left\vert M_{BA}^{\frac{1}{\epsilon}} \right\vert ~. \label{eqSep3.1}%
\end{equation}
If their symplectic eigenvalues satisfy
\begin{equation}
\lambda_{\sigma_{A},j} \left(  \widetilde{M}_{AA}^{\epsilon} \right)  \leq1
\text{ and } \lambda_{\sigma_{B},k} \left(  \widetilde{M}_{BB}^{\frac
{1}{\epsilon}} \right)  ~, \label{eqSep3.2}%
\end{equation}
for all $j=1, \cdots, n_{A}$ and all $k=1, \cdots, n_{B}$, then the Gaussian
state $\widehat{\rho}$ with covariance ellipsoid (\ref{eqSep2}) is separable.
\end{theorem}

\begin{proof}
We proceed as in the previous proofs and apply this time the inequality
(\ref{eqImp22}) for the set of positive numbers $\epsilon=\left(  \epsilon
_{1}, \cdots, \epsilon_{2n_{A}}\right)  \in\mathbb{R}_{+}^{2 n_{A}}$.
\begin{equation}%
\begin{array}
[c]{c}%
z_{A} \cdot M_{AB} z_{B} \leq\sum_{j=1}^{2 n_{A}} \mu_{j}^{AB} \left\vert
u_{A,j} \right\vert ~\left\vert v_{B,j} \right\vert \\
\\
\leq\sum_{j=1}^{2 n_{A}} \mu_{j}^{AB} \left(  \frac{\epsilon_{j} \left\vert
u_{A,j} \right\vert ^{2}}{2} +\frac{\left\vert v_{B,j} \right\vert ^{2}%
}{2\epsilon_{j}}\right)  = \left\vert M_{AB}^{\epsilon} \right\vert z_{A}^{2}+
\left\vert M_{BA}^{\frac{1}{\epsilon}} \right\vert z_{B}^{2}~.
\end{array}
\label{eqSep3.4}%
\end{equation}
It follows that:
\begin{equation}%
\begin{array}
[c]{c}%
Mz^{2}= M_{AA}z_{A}^{2}+ 2 z_{A} \cdot M_{AB} z_{B} + M_{BB}z_{B}^{2}\\
\\
\leq\left(  M_{AA} + \left\vert M_{AB}^{\epsilon}\right\vert \right)
z_{A}^{2} + \left(  M_{BB}+ \left\vert M_{BA}^{\frac{1}{\epsilon}} \right\vert
\right)  z_{B}^{2}~,
\end{array}
\label{eqSep3.5}%
\end{equation}
which means that:
\begin{equation}
M \leq\widetilde{M}_{AA}^{\epsilon} \oplus\widetilde{M}_{BB}^{\frac
{1}{\epsilon}}~. \label{eqSep3.6}%
\end{equation}
The rest follows as previously.
\end{proof}

\subsection{Another example of non-necessity}

\label{Section4.6}

We will now show, with a particular example when $n_{A}=n_{B}=1$, that the
criterion stated in Theorem \ref{TheoremSep2} is not necessary for separability.

Let $M$ be given by:
\begin{equation}
M=\left(
\begin{array}
[c]{cccc}%
\frac{2}{3} & 0 & \frac{1}{2} & 0\\
0 & \frac{2}{3} & 0 & \frac{1}{2}\\
\frac{1}{2} & 0 & \frac{1}{8} & 0\\
0 & \frac{1}{2} & 0 & \frac{1}{8}%
\end{array}
\right)  \label{eqImp29}%
\end{equation}
With the previous notation, we have;
\begin{equation}%
\begin{array}
[c]{l}%
M_{AA}=\left(
\begin{array}
[c]{cc}%
\frac{2}{3} & 0\\
0 & \frac{2}{3}%
\end{array}
\right)  \text{ \ },\text{ \ }M_{BB}=\left(
\begin{array}
[c]{cc}%
\frac{1}{8} & 0\\
0 & \frac{1}{8}%
\end{array}
\right) \\
\\
M_{AB}=M_{BA}=\left\vert M_{AB}\right\vert =\left\vert M_{BA}\right\vert
=\left(
\begin{array}
[c]{cc}%
\frac{1}{2} & 0\\
0 & \frac{1}{2}%
\end{array}
\right)
\end{array}
~. \label{eqImp30}%
\end{equation}
We thus have:
\begin{equation}
M_{AA}^{\sharp}=\left(
\begin{array}
[c]{cc}%
\frac{7}{6} & 0\\
0 & \frac{7}{6}%
\end{array}
\right)  ,\hspace{0.5cm}M_{BB}^{\sharp}=\left(
\begin{array}
[c]{cc}%
\frac{5}{8} & 0\\
0 & \frac{5}{8}%
\end{array}
\right)  ~, \label{eqImp31}%
\end{equation}
which entails that $\lambda_{\sigma_{A}}\left(  M_{AA}^{\sharp}\right)
=\frac{7}{6}>1$, while $\lambda_{\sigma_{B}}\left(  M_{BB}^{\sharp}\right)
=\frac{5}{8}<1$. We conclude that the condition in the criterion of Theorem
\ref{TheoremSep2} is not respected. However, the Gaussian state associated
with $M$ is a separable state. Indeed, we have for $\varepsilon=\left(
\frac{13}{21},\frac{13}{21}\right)  $ and $\frac{1}{\varepsilon}=\left(
\frac{21}{13},\frac{21}{13}\right)  $:
\begin{equation}
\widetilde{M}_{AA}^{\varepsilon}=\left(
\begin{array}
[c]{cc}%
\frac{41}{42} & 0\\
0 & \frac{41}{42}%
\end{array}
\right)  ~,\hspace{1cm}\widetilde{M}_{BB}^{\frac{1}{\varepsilon}}=\left(
\begin{array}
[c]{cc}%
\frac{95}{104} & 0\\
0 & \frac{95}{104}%
\end{array}
\right)  ~. \label{eqImp32}%
\end{equation}
We thus have:
\begin{equation}
\lambda_{\sigma_{A}}\left(  \widetilde{M}_{AA}^{\varepsilon}\right)
=\frac{41}{42}<1~,\hspace{0.5cm}\lambda_{\sigma_{B}}\left(  \widetilde{M}%
_{BB}^{\frac{1}{\varepsilon}}\right)  =\frac{95}{104}<1~. \label{eqImp33}%
\end{equation}
From Theorem \ref{TheoremSep3}, we conclude that the associated Gaussian state
is separable.

\subsection{The fourth separability criterion: a particular case}

In this section, we derive another sufficient criterion, which applies only to
the particular case where all the blocks, $M_{AA}$, $M_{AB}$ and $M_{BB}$ are
either diagonal or can be brought to a diagonal form by a symplectic
transformation $S_{A} \oplus S_{B}$. We illustrate this example with several
pictures which highlight the geometric nature of the problem.

We will thus assume that there exist $S_{A}\in\operatorname*{Sp}(n_{A})$ and
$S_{B}\in\operatorname*{Sp}(n_{B})$, such that, for $S=S_{A}\oplus S_{B}$:
\begin{equation}
SMS^{T}=M_{D}=\left(
\begin{array}
[c]{cc}%
A & D\\
D^{T} & B
\end{array}
\right)  ~,\label{eqSep4.1}%
\end{equation}
where we have the following diagonal blocks:
\begin{equation}
A=\operatorname*{diag}\left(  \Lambda_{A},\Lambda_{A}\right)
,~B=\operatorname*{diag}\left(  \Lambda_{B},\Lambda_{B}\right)
~,\label{eqSep4.2}%
\end{equation}
with
\begin{equation}%
\begin{array}
[c]{l}%
\Lambda_{A}=\operatorname*{diag}\left(  \lambda_{\sigma_{A},1}(M_{AA}%
),...,\lambda_{\sigma_{A},n_{A}}(M_{AA})\right)  \\
\\
\Lambda_{B}=\operatorname*{diag}\left(  \lambda_{\sigma_{B},1}(M_{BB}%
),...,\lambda_{\sigma_{B},n_{B}}(M_{BB})\right)  ~,
\end{array}
\label{eqSep4.3}%
\end{equation}
and $D$ is a $2n_{A}\times2n_{B}$ matrix of the form:
\begin{equation}
D=\left(
\begin{array}
[c]{cc}%
E & 0_{n_{A}\times n_{B}}\\
0_{n_{A}\times n_{B}} & F
\end{array}
\right)  \label{eqMatrixD1}%
\end{equation}
where $E$ and $F$ are diagonal $n_{A}\times n_{B}$ matrices with entries:
\begin{equation}
E_{j,k}=d_{j}\delta_{j,k},\quad F_{j,k}=d_{j+n_{A}}\delta_{j,k},\quad
j=1,\cdots,n_{A}~,~k=1,\cdots,n_{B}~.\label{eqMatrixD2}%
\end{equation}

In the sequel, we will need to consider the following $2 \times2$ matrices for
a set of numbers $a_{j},b_{j}$:
\begin{equation}
Q_{j}(a_{j},b_{j})= \left(
\begin{array}
[c]{cc}%
a_{j}- \lambda_{\sigma_{A},j} (M_{AA}) & -d_{j}\\
- d_{j} & b_{j}- \lambda_{\sigma_{B},j} (M_{BB})
\end{array}
\right)  , j=1,...,n_{A} \label{eqSep4.5}%
\end{equation}
and
\begin{equation}
P_{j}(a_{j},b_{j})= \left(
\begin{array}
[c]{cc}%
\frac{1}{a_{j}}- \lambda_{\sigma_{A},j} (M_{AA}) & -d_{j+n_{A}}\\
- d_{j+n_{A}} & \frac{1}{b_{j}}- \lambda_{\sigma_{B},j} (M_{BB})
\end{array}
\right)  , j=1,...,n_{A} \label{eqSep4.6}%
\end{equation}

\begin{theorem}
\label{TheoremSep4} Suppose that there exist a set of numbers $a_{1}, \cdots,
a_{n_{A}} >0$ and $b_{1}, \cdots, b_{n_{B}}>0$, such that:
\begin{equation}%
\begin{array}
[c]{l}%
\lambda_{\sigma_{A},j}(M_{AA}) \leq a_{j} \leq\frac{1}{\lambda_{\sigma_{A}%
,j}(M_{AA})}, ~j=1, \cdots, n_{A}\\
\\
\lambda_{\sigma_{B},k}(M_{BB}) \leq b_{k} \leq\frac{1}{\lambda_{\sigma_{B}%
,k}(M_{BB})}, ~k=1, \cdots, n_{B}%
\end{array}
\label{eqSep4.7}%
\end{equation}
and
\begin{equation}
\det Q_{j}(a_{j},b_{j}) \geq0 , ~\det P_{j}(a_{j},b_{j}) \geq0 , ~ j=1,
\cdots, n_{A} ~. \label{eqSep4.8}%
\end{equation}
Then the Gaussian state with covariance ellipsoid (\ref{eqSep2}) is separable.
\end{theorem}

\begin{proof}
First of all, notice that if the state is separable, then there exist
$S_{A}^{\prime} \in Sp(n_{A})$ and $S_{B}^{\prime} \in Sp(n_{B})$, such that:
\begin{equation}%
\begin{array}
[c]{c}%
M \leq\left(  \left(  S_{A}^{\prime}\right)  ^{T} S_{A}^{\prime} \right)
\oplus\left(  \left(  S_{B}^{\prime}\right)  ^{T} S_{B}^{\prime} \right) \\
\\
\Leftrightarrow SM S^{T}\leq\left(  \left(  S_{A}^{\prime}S_{A}^{T}\right)
^{T} \left(  S_{A}^{\prime} S_{A}^{T}\right)  \right)  \oplus\left(  \left(
S_{B}^{\prime}S_{B}\right)  ^{T} \left(  S_{B}^{\prime}S_{B}^{T} \right)
\right)  ~.
\end{array}
\label{eqSep4.9}%
\end{equation}
Thus, $\widehat{\rho}$ is separable if and only if the Gaussian state with
covariance ellipsoid given by the matrix $M_{D}=SMS^{T}$ is separable. We may
therefore assume that $M$ is of the form (\ref{eqSep4.1})-(\ref{eqMatrixD2}).

Next, consider the positive symplectic matrix $P_{A}\oplus P_{B}$, with
\begin{equation}%
\begin{array}
[c]{l}%
P_{A}=\operatorname*{diag}\left(  a_{1},\cdots,a_{n_{A}},\frac{1}{a_{1}%
},\cdots,\frac{1}{a_{n_{A}}}\right)  ~,\\
\\
P_{B}=\operatorname*{diag}\left(  b_{1},\cdots,b_{n_{B}},\frac{1}{b_{1}%
},\cdots,\frac{1}{b_{n_{B}}}\right)  ~.
\end{array}
\label{eqSep4.10}%
\end{equation}
If $M\leq P_{A}\oplus P_{B}$, then $\widehat{\rho}$ is a separable state.
Writing $z=(z_{A},z_{B})$ and $z_{A}=\left(  x_{A},p_{A}\right)  $,
$z_{B}=\left(  x_{B},p_{B}\right)  $, this is equivalent to:
\begin{equation}%
\begin{array}
[c]{c}%
\sum_{j=1}^{n_{A}}\lambda_{\sigma_{A},j}(M_{AA})\left(  x_{A,j}^{2}%
+p_{A,j}^{2}\right)  +2\sum_{j=1}^{2n_{A}}d_{j}z_{A,j}z_{B,j}+\\
\\
+\sum_{j=1}^{n_{B}}\lambda_{\sigma_{B},j}(M_{BB})\left(  x_{B,j}^{2}%
+p_{B,j}^{2}\right) \\
\\
\leq\sum_{j=1}^{n_{A}}\left(  a_{j}x_{A,j}^{2}+\frac{p_{A,j}^{2}}{a_{j}%
}\right)  +\sum_{j=1}^{n_{B}}\left(  b_{j}x_{B,j}^{2}+\frac{p_{B,j}^{2}}%
{b_{j}}\right)
\end{array}
\label{eqSep4.11}%
\end{equation}
These equations can be decoupled for each $j$ and we obtain the set of
inequalities:
\begin{equation}%
\begin{array}
[c]{c}%
\lambda_{\sigma_{A},j}(M_{AA})x_{A,j}^{2}+2d_{j}x_{A,j}x_{B,j}+\lambda
_{\sigma_{B},j}(M_{BB})x_{B,j}^{2}\\
\\
\leq a_{j}x_{A,j}^{2}+b_{j}x_{B,j}^{2}~,j=1,\cdots,n_{A}~,
\end{array}
\label{eqSep4.12}%
\end{equation}

\begin{equation}%
\begin{array}
[c]{c}%
\lambda_{\sigma_{A},j}(M_{AA})p_{A,j}^{2}+2d_{n_{A}+j}p_{A,j}p_{B,j}%
+\lambda_{\sigma_{B},j}(M_{BB})p_{B,j}^{2}\\
\\
\leq\frac{p_{A,j}^{2}}{a_{j}}+\frac{p_{B,j}^{2}}{b_{j}}~,j=1,\cdots,n_{A}~,
\end{array}
\label{eqSep4.13}%
\end{equation}
and
\begin{equation}
\lambda_{\sigma_{B},j}(M_{BB})\left(  x_{B,j}^{2}+p_{B,j}^{2}\right)  \leq
b_{j}x_{B,j}^{2}+\frac{p_{B,j}^{2}}{b_{j}}~,j=n_{A}+1,\cdots,n_{B}%
~.\label{eqSep4.14}%
\end{equation}
Inequalities (\ref{eqSep4.12})-(\ref{eqSep4.14}) are equivalent to
(\ref{eqSep4.7}) and (\ref{eqSep4.8}).
\end{proof}

If $n_{A}=n_{B}$, then we can discard inequalities (\ref{eqSep4.14}). If
$n_{B} >n_{A}$, then we just have to find $b_{n_{A}+1}, \cdots, b_{n_{B}}>0$,
such that $\lambda_{\sigma_{B},j}(M_{BB}) \leq b_{j} \leq\frac{1}%
{\lambda_{\sigma_{B},j}(M_{BB})}, ~j=n_{A}+1, \cdots, n_{B}$. The nontrivial
part corresponds to determining the remaining constants, $a_{j},b_{j}$ for
$j=1, \cdots, n_{A}$. In the general case, these are easily obtained numerically.

Each solution $a_{j},b_{k}$, $j=1,...,n_{A}$ and $k=1,...,n_{B}$ determines an
ellipsoid
\begin{equation}
\label{ellAB}\Omega_{AB}=\{z\in\mathbb{R}^{2n}: P_{A} z_{A}^{2}+ P_{B}
z_{B}^{2} \le\hbar\} \subset\Omega_{D}%
\end{equation}
where $P_{A},P_{B}$ are given by (\ref{eqSep4.10}), and $\Omega_{D}$ is the
covariant ellipsoid of the matrix $M_{D}$. The projection of $\Omega_{AB}$
onto the plane $x_{A,j}x_{B,j}$ determines an ellipse (of size $1/\sqrt{a_{j}%
},1/\sqrt{b_{j}}$, if we assume $\hbar=1$) and the projection onto the plane
$p_{A,j}p_{B,j}$ determines another ellipse, "conjugate" to the first one, and
of size $\sqrt{a_{j}},\sqrt{b_{j}}$). These two ellipses are enclosed in the
projections of $\Omega_{D}$ onto these two planes. We also conclude from
(\ref{eqSep4.7}) that (cf.(\ref{eqSep4.1},\ref{eqSep4.10})):
\[
P_{A} \ge A \ge M_{D} / B = A - D B^{-1} D^{T}
\]
and so $\Pi_{A} \Omega_{AB} \subset\Pi_{A} \Omega_{D}$. An equivalent result
is valid for the projection $\Pi_{B}$. These geometrical relations are
illustrated by the example at the end of this section.

A set of conditions equivalent to those of Theorem \ref{TheoremSep4} is the
following. We use the abbreviated notation $\lambda_{j}^{A}= \lambda
_{\sigma_{A},j} (M_{AA})$, $\lambda_{j}^{B}= \lambda_{\sigma_{B},j} (M_{BB})$.

\begin{lemma}
\label{LemmaSep1} The following set of conditions are equivalent.

\vspace{0.2 cm}

\begin{enumerate}
\item The matrices $Q_{j}(a,b)$ and $P_{j}(a,b)$ are positive semi-definite
for some $a,b>0$.

\vspace{0.2 cm}

\item There exists $a_{0}\in\left[  \lambda_{j}^{A},\frac{1}{\lambda_{j}^{A}%
}\right]  $, such that $f(a_{0})\geq0$, where $f(x)=\alpha x^{2}+\beta
x+\gamma$, with:
\begin{equation}%
\begin{array}
[c]{l}%
\alpha=\lambda_{j}^{A}\lambda_{j}^{B}-\lambda_{j}^{B}d_{n_{A}+j}^{2}%
-\lambda_{j}^{A}\\
\\
\beta=1+\left(  \lambda_{j}^{A}\right)  ^{2}-\left(  \lambda_{j}^{B}\right)
^{2}+\left(  \lambda_{j}^{A}\lambda_{j}^{B}-d_{j}^{2}\right)  \cdot\left(
d_{n_{A}+j}^{2}-\lambda_{j}^{A}\right)  \\
\\
\gamma=\left(  \lambda_{j}^{A}\lambda_{j}^{B}-d_{j}^{2}\right)  \lambda
_{j}^{B}~.
\end{array}
\label{eqSep4.15}%
\end{equation}

\end{enumerate}
\end{lemma}

\begin{proof}
For simplicity, we write $\lambda_{j}^{A}=\lambda^{A}$, $\lambda_{j}%
^{B}=\lambda^{B}$, $d_{j}=d$ and $d_{n_{A}+j}=D$. Conditions 1 are equivalent
to
\begin{equation}
\lambda^{A}\leq a\leq\frac{1}{\lambda^{A}}~,\hspace{0.5cm}\lambda^{B}\leq
b\leq\frac{1}{\lambda^{B}}~, \label{eqSep4.16}%
\end{equation}
and
\begin{equation}
(a-\lambda^{A})\cdot(b-\lambda^{B})\geq d^{2}~,\hspace{0.5cm}\left(  \frac
{1}{a}-\lambda^{A}\right)  \cdot\left(  \frac{1}{b}-\lambda^{B}\right)  \geq
D^{2}~. \label{eqSep4.17}%
\end{equation}
From the first inequality in (\ref{eqSep4.17}), we obtain:
\begin{equation}
\lambda^{B}+\frac{d^{2}}{a-\lambda^{A}}\leq b~. \label{eqSep4.18}%
\end{equation}
Similarly, from the second inequality, we obtain:
\begin{equation}
b\leq\frac{1}{\lambda^{B}+\frac{D^{2}}{\frac{1}{a}-\lambda^{B}}}~.
\label{eqSep4.19}%
\end{equation}
If $\lambda^{A}\leq a\leq\frac{1}{\lambda^{A}}$, then we conclude from
(\ref{eqSep4.18}) and (\ref{eqSep4.19}) that we have automatically
$\lambda^{B}\leq b\leq\frac{1}{\lambda^{B}}$. It follows that conditions 1 are
equivalent to $\lambda^{A}\leq a\leq\frac{1}{\lambda^{A}}$ and%
\begin{equation}
\lambda^{B}+\frac{d^{2}}{a-\lambda^{A}}\leq\frac{1}{\lambda^{B}+\frac{D^{2}%
}{\frac{1}{a}-\lambda^{B}}}\Leftrightarrow f(a)\geq0~, \label{eqSep4.20}%
\end{equation}
which concludes the proof.
\end{proof}

As an example let us consider the case $n_{A}=n_{B}=1$ with the matrix
\begin{equation}
M=\left(
\begin{array}
[c]{cccc}%
\frac{1}{2} & 0 & \frac{2}{3} & 0\\
0 & \frac{1}{2} & 0 & \frac{1}{4}\\
\frac{2}{3} & 0 & \frac{17}{18} & 0\\
0 & \frac{1}{4} & 0 & \frac{3}{16}%
\end{array}
\right)  ~.\label{eqSep4.21}%
\end{equation}
The associated matrix $M_{D}$ is:
\begin{equation}
M_{D}=\left(
\begin{array}
[c]{cccc}%
\frac{1}{2} & 0 & \left(  \frac{2}{51}\right)  ^{1/4} & 0\\
0 & \frac{1}{2} & 0 & \frac{(17/54)^{1/4}}{2}\\
\left(  \frac{2}{51}\right)  ^{1/4} & 0 & \frac{\sqrt{17/6}}{4} & 0\\
0 & \frac{(17/54)^{1/4}}{2} & 0 & \frac{\sqrt{17/6}}{4}%
\end{array}
\right)  \label{eqSep4.22}%
\end{equation}
where we used the fact that $\lambda_{\sigma_{A},1}(M_{AA})=1/2$ and
$\lambda_{\sigma_{B},1}(M_{BB})=\frac{\sqrt{17/6}}{4}$. The matrices
$Q_{1},P_{1}$ for this case are:
\begin{align*}
Q_{1}(a,b)  & =\left(
\begin{array}
[c]{cc}%
a-1/2 & \left(  \frac{2}{51}\right)  ^{1/4}\\
\left(  \frac{2}{51}\right)  ^{1/4} & b-\frac{\sqrt{17/6}}{4}%
\end{array}
\right)  \\
P_{1}(a,b)  & =\left(
\begin{array}
[c]{cc}%
1/a-1/2 & \frac{(17/54)^{1/4}}{2}\\
\frac{(17/54)^{1/4}}{2} & 1/b-\frac{\sqrt{17/6}}{4}%
\end{array}
\right)
\end{align*}
We are looking for solutions of
\begin{equation}
\det Q_{1}(a,b)\geq0\quad,\quad\det P_{1}(a,b)\geq0\label{eqSep4.23}%
\end{equation}
in the range
\[
1/2\leq a\leq2\quad\mbox{and}\quad\frac{\sqrt{17/6}}{4}\leq b\leq\frac
{4}{\sqrt{17/6}}\,.
\]
These solutions can be obtained numerically. They are given by the points
between the two curves in Figure 1.

{\normalsize \vspace{0.0cm} {\normalsize \begin{figure}[h]
\includegraphics [scale=0.3] {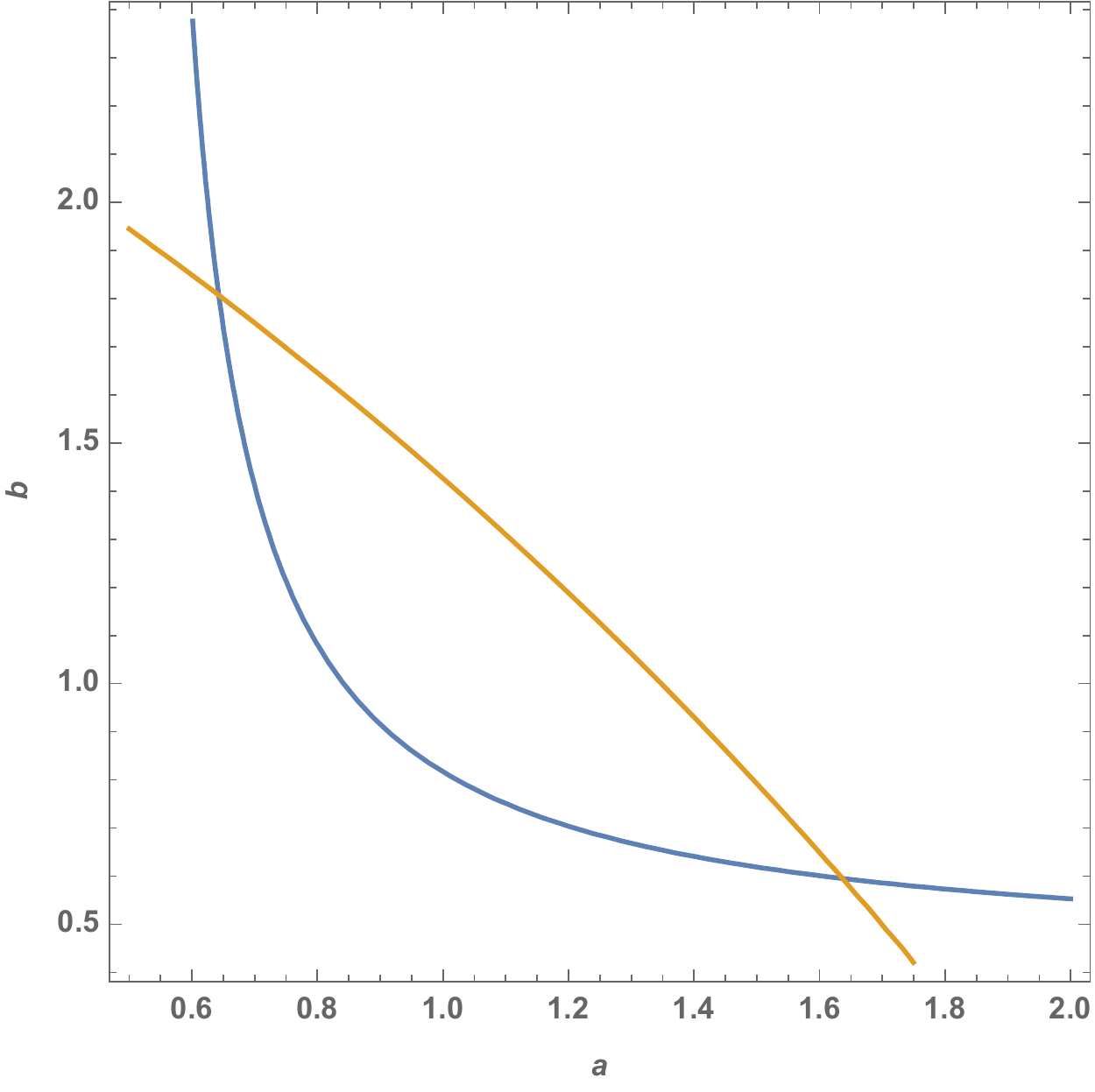}\end{figure}\vspace{0.0cm}    }
{\footnotesize {Figure 1: The solutions of eq. (\ref{eqSep4.23}) are given by
the points between the two curves. Each point corresponds to an ellipsoid
$\Omega_{AB}$ (\ref{ellAB}) that is contained in the covariant ellipsoid
$\Omega_{D}$ associated to (\ref{eqSep4.22})}. \vspace{0.5cm} }}

{\normalsize Recall that $\Omega_{AB}$ is given by (\ref{ellAB}) and that
$\Omega_{D}$ is the covariant ellipsoid associated to $M_{D}$. In Figures 2.1
to 2.4 we consider the case $a=1.6$, $b=0.6$, $\hbar=1$ and plot the
projections of $\Omega_{D}$ and $\Omega_{AB}$ onto the planes $x_{A,1}p_{A,1}%
$, $x_{B,1}p_{B,1}$, $x_{A,1}x_{B,1}$ and $p_{A,1}p_{B,1}$. }

{\normalsize \begin{figure}[h]
{\normalsize \center\includegraphics [scale=0.3] {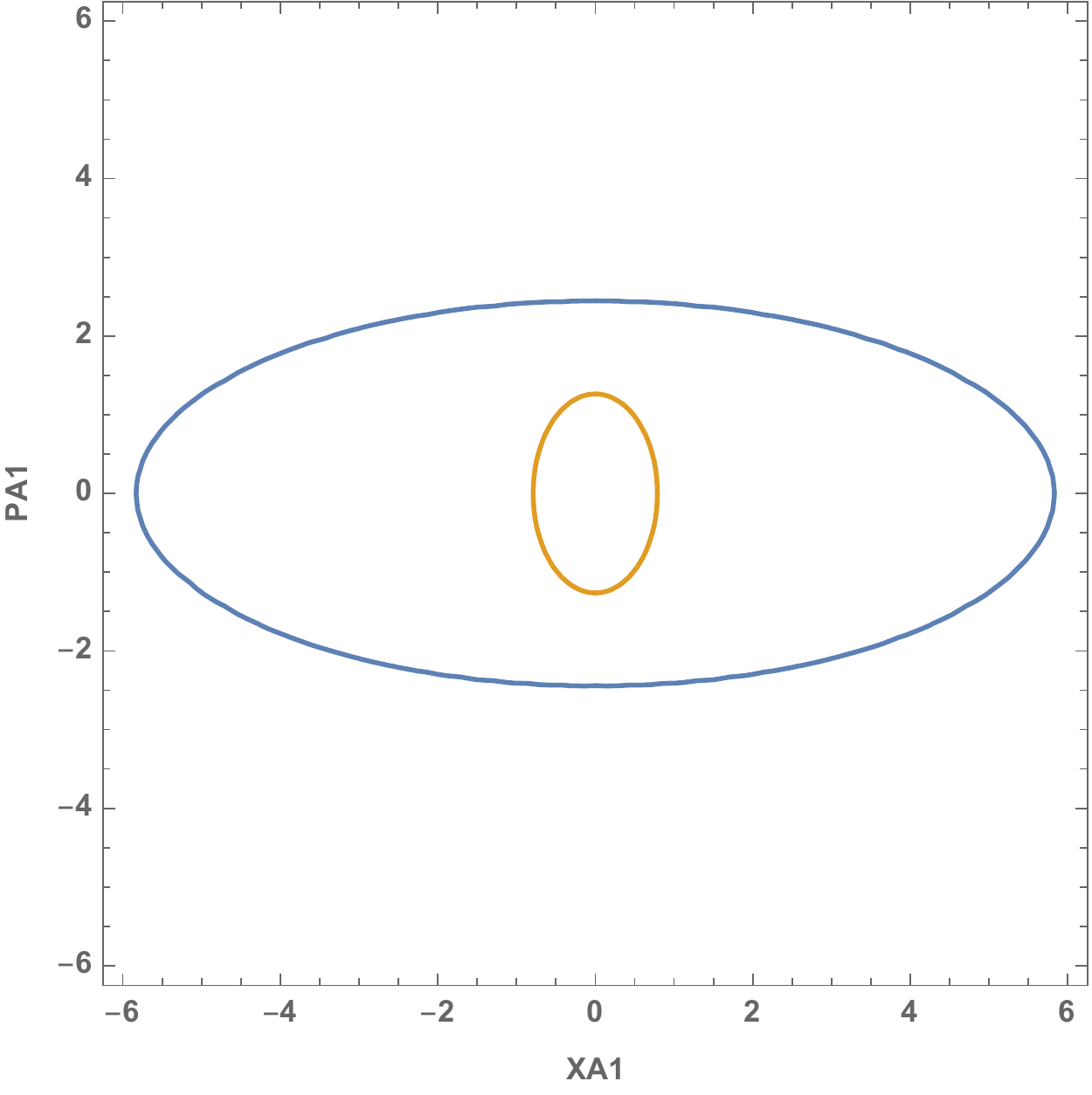} \hspace{1cm}
\includegraphics [scale=0.3] {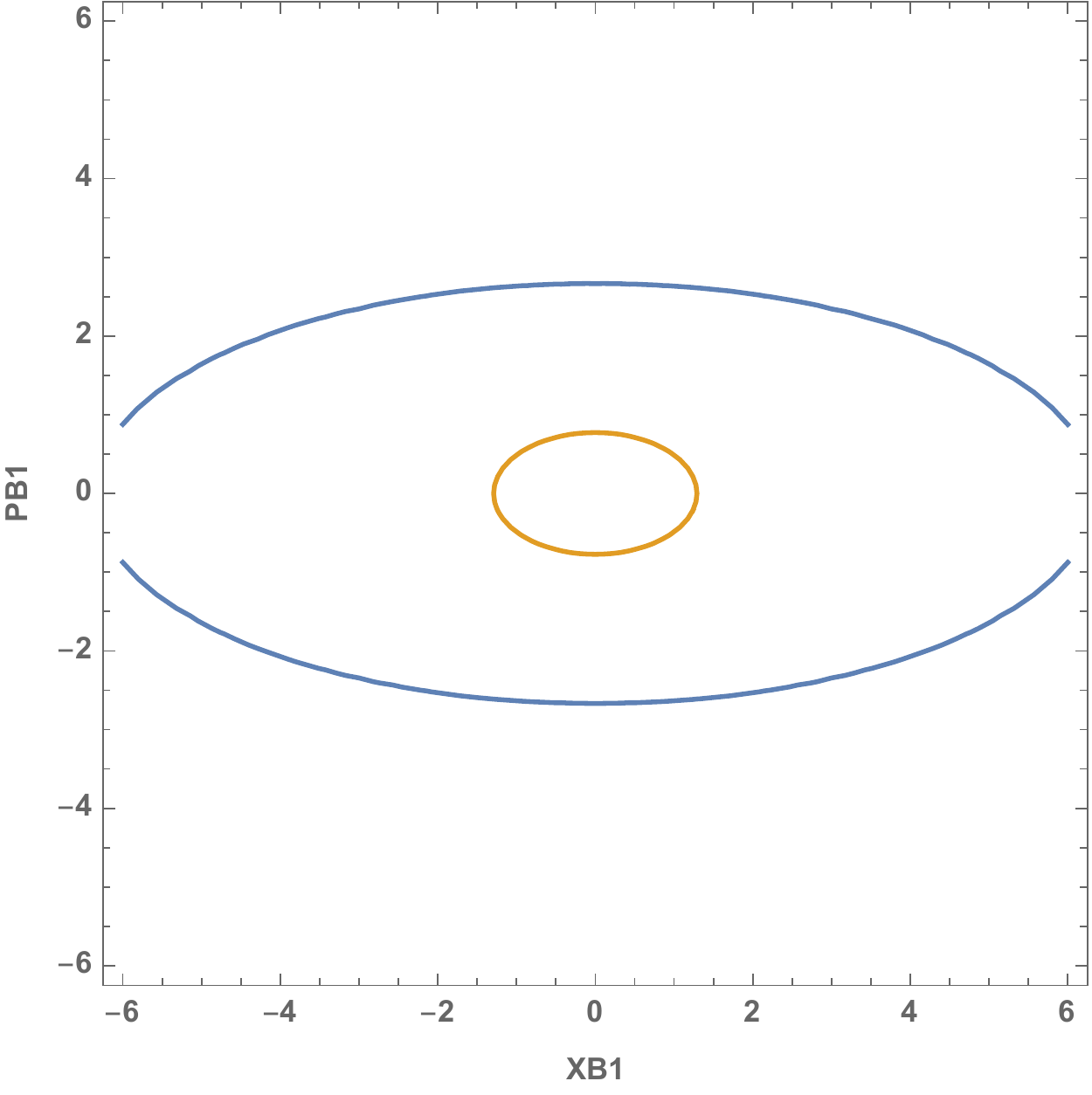} }\end{figure}\vspace{0.0cm}  }

{\normalsize {\footnotesize {Figures 2.1 and 2.2: Projections of $\Omega_{D}$
and $\Omega_{AB}$ onto the $x_{A,1}p_{A,1}$ plane (left) and onto the
$x_{B,1}p_{B,1}$ plane (right) for the case $a=1.6$, $b=0.6$ and $\hbar=1$.}
\vspace{3.1cm} } }

{\normalsize \begin{figure}[h]
{\normalsize \center\includegraphics [scale=0.3] {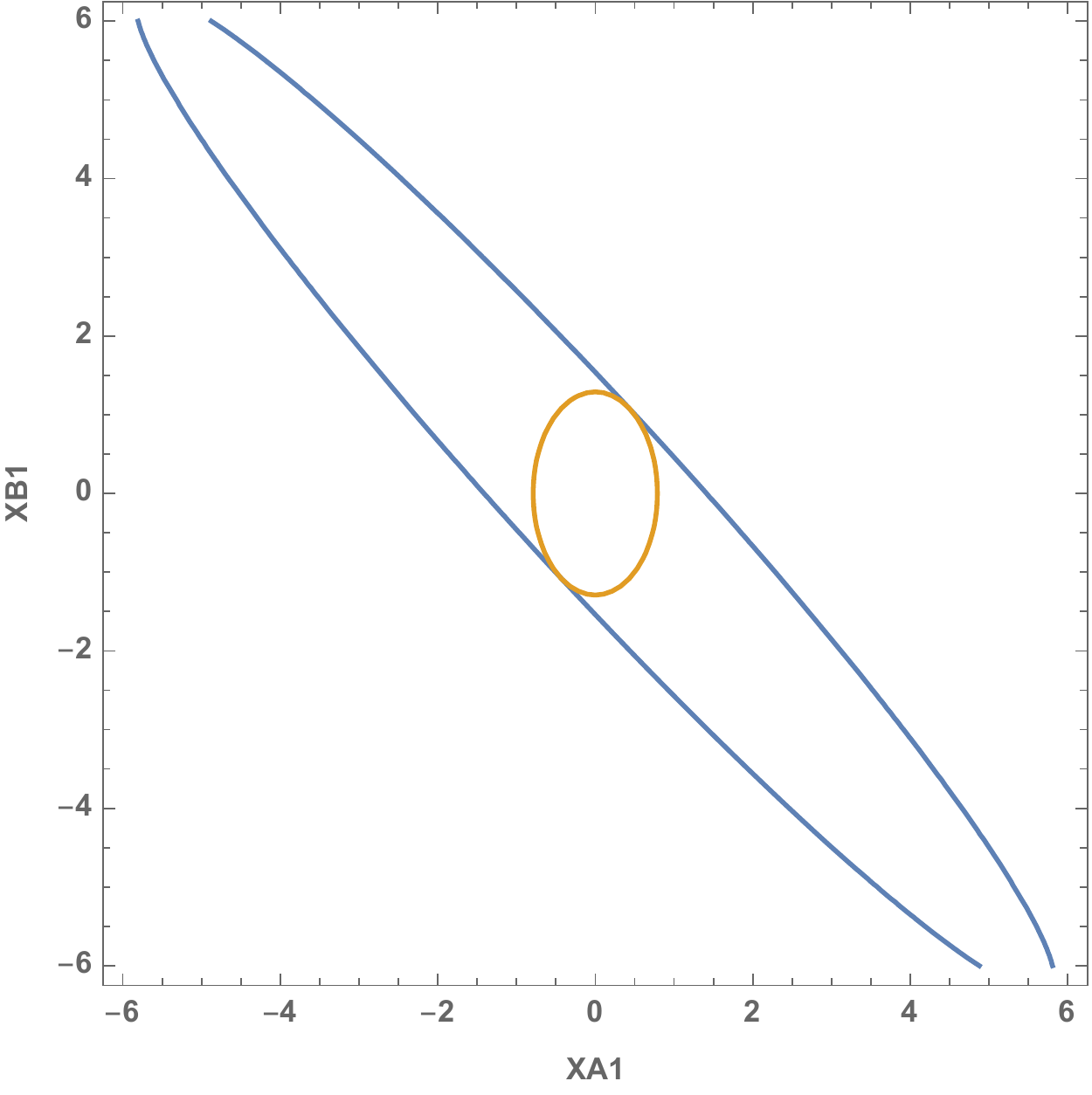} \hspace{1cm}
\includegraphics [scale=0.3] {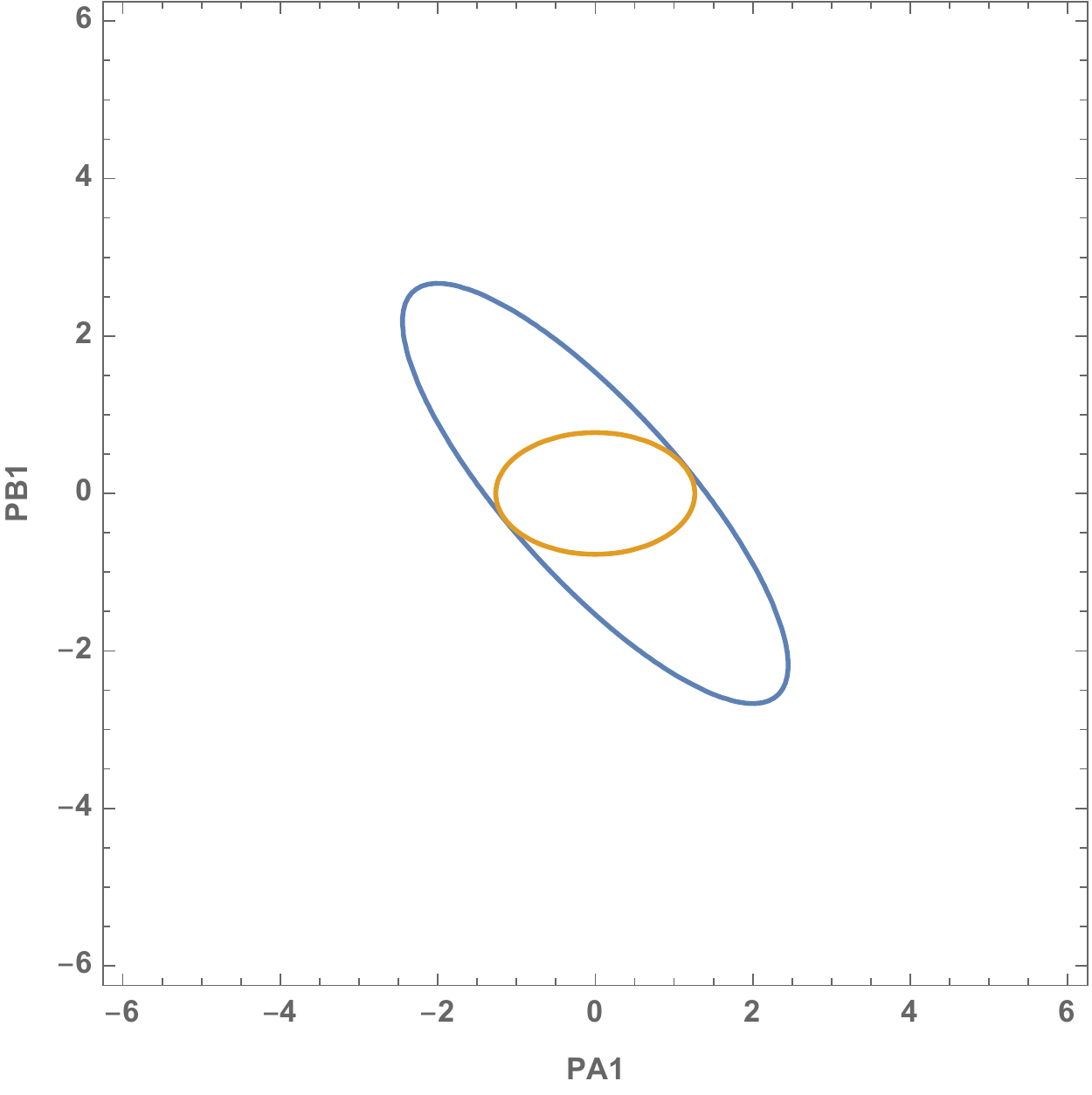} }\end{figure}\vspace{0.0cm}  }

{\normalsize {\footnotesize {Figures 2.3 and 2.4: Projections of $\Omega_{D}$
and $\Omega_{AB}$ onto the $x_{A,1}x_{B,1}$ plane and onto the $p_{A,1}%
p_{B,1}$ plane for the case $a=1.6$, $b=0.6$ and $\hbar=1$.} \vspace{0.1cm} }
}

{\normalsize \vspace{0.25cm} In Figures 3.1 to 3.4 the plots represent the
same projections of $\Omega_{D}$ and $\Omega_{AB}$ for another solution of
(\ref{eqSep4.23}): $a=0.7$ and $b=1.8$.  }

{\normalsize \begin{figure}[h]
{\normalsize \center\includegraphics [scale=0.3] {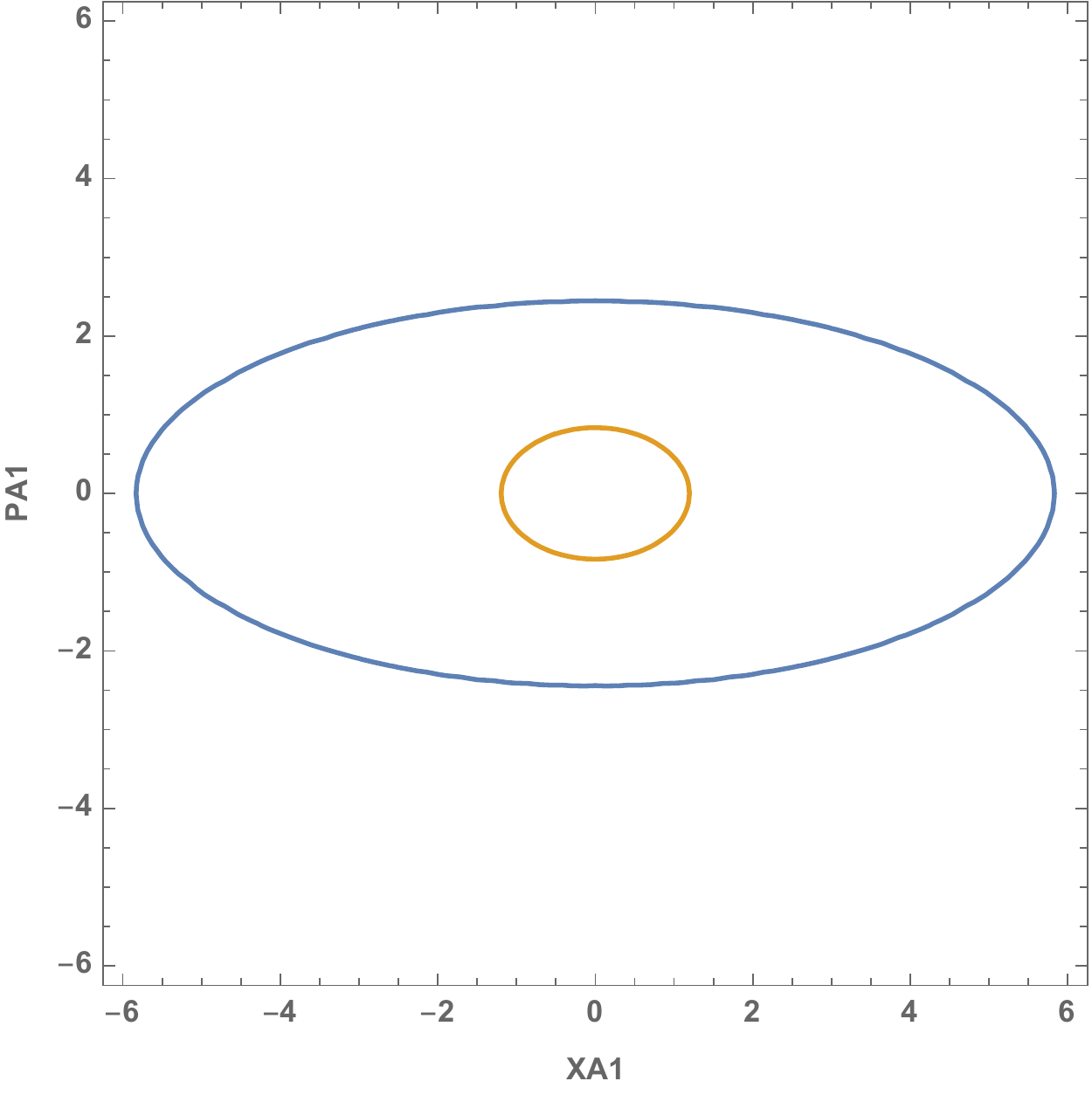} \hspace{1cm}
\includegraphics [scale=0.3] {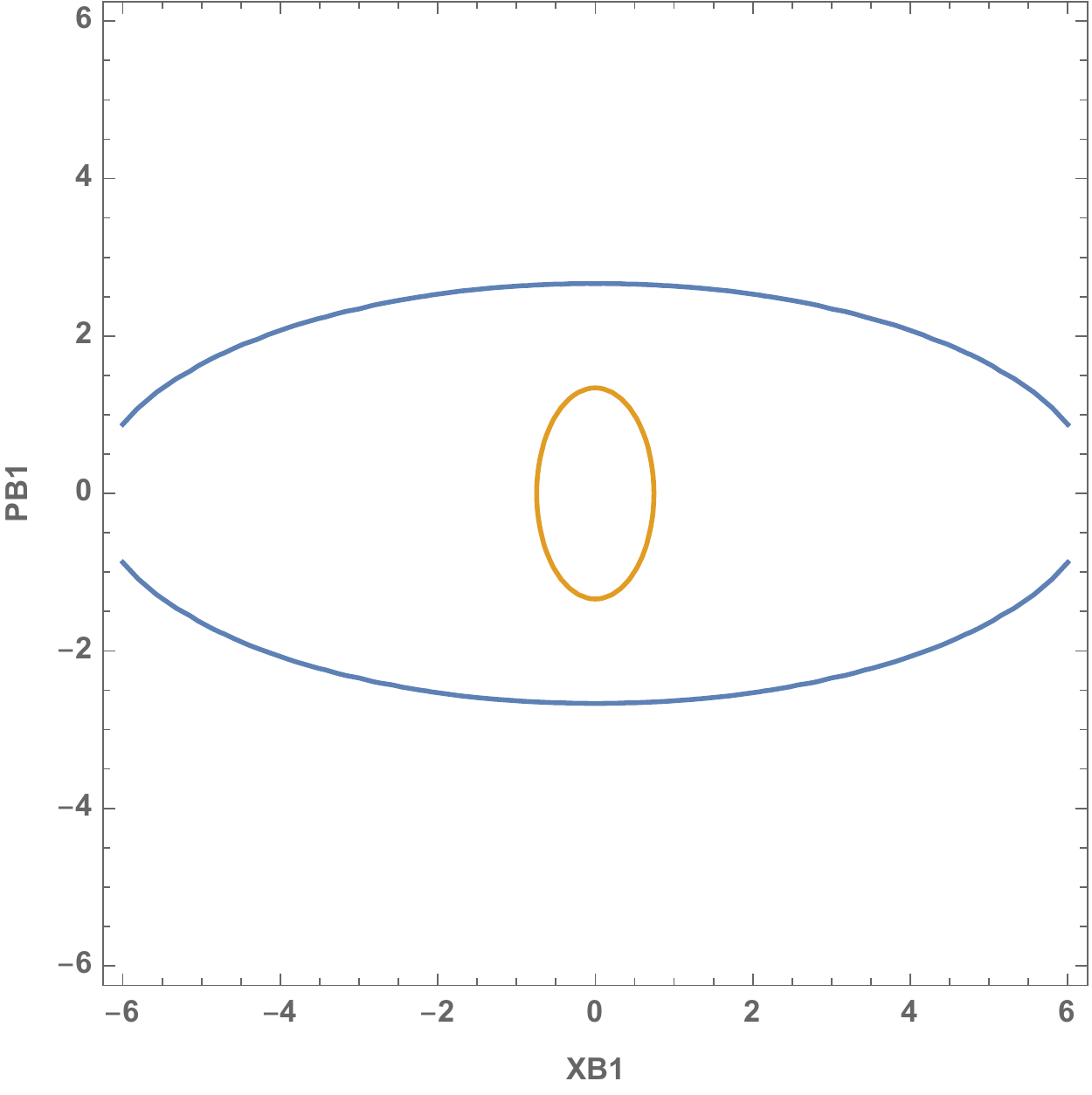} }\end{figure}\vspace{0.0cm}  }

{\normalsize \begin{figure}[h]
{\normalsize \center\includegraphics [scale=0.3] {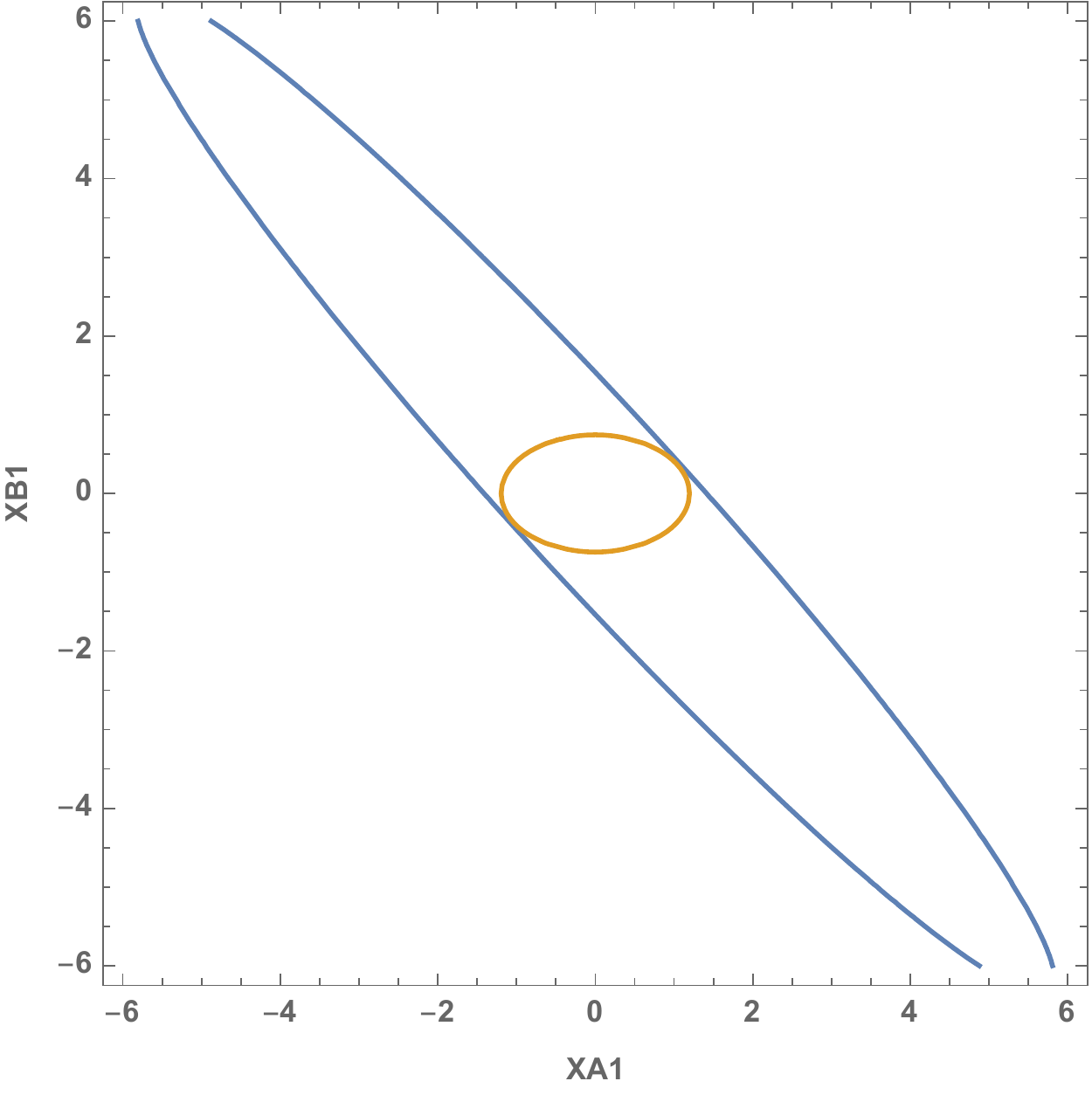} \hspace{1cm}
\includegraphics [scale=0.3] {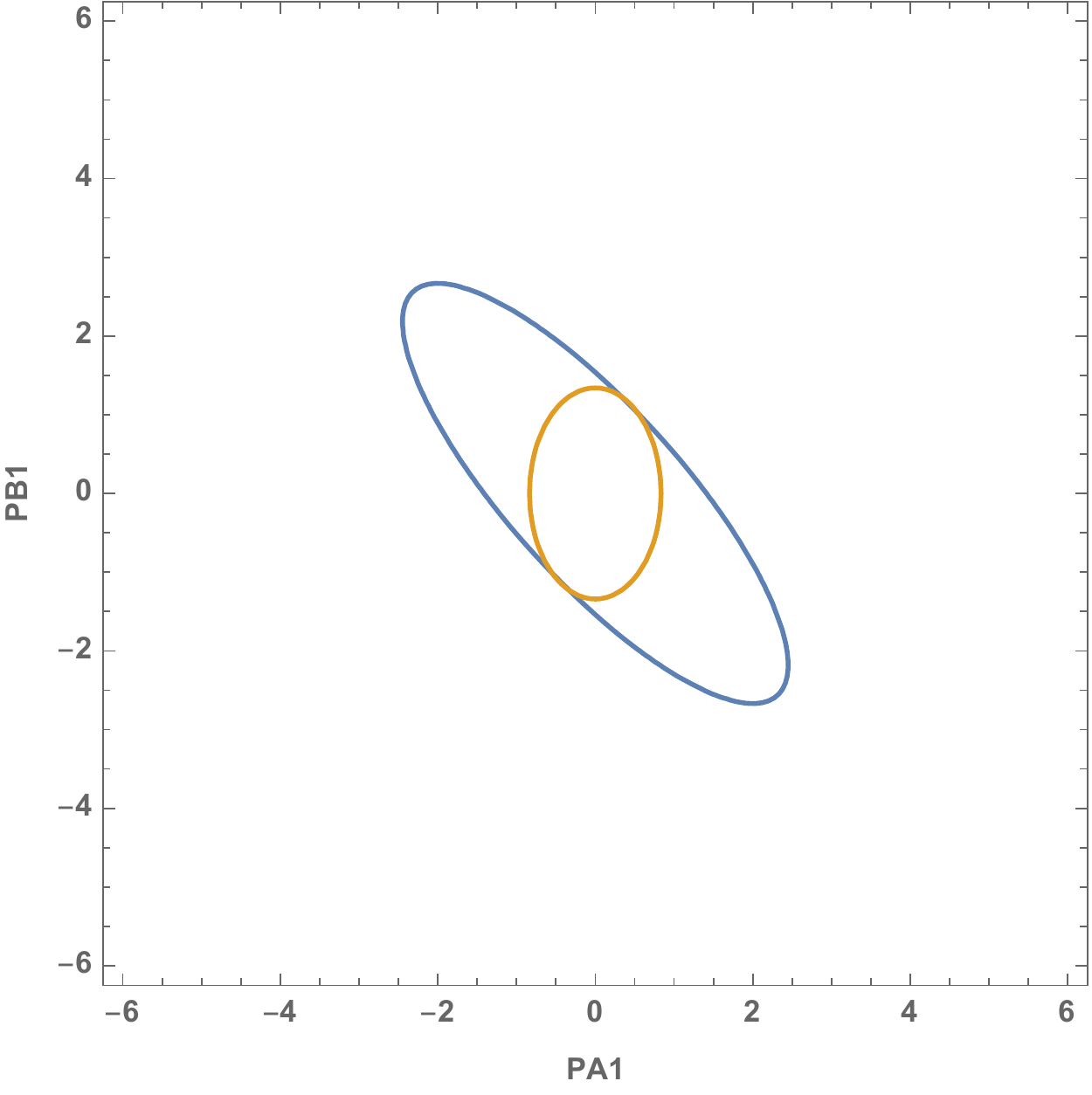} }\end{figure}\vspace{0.0cm}  }

{\normalsize {\footnotesize {Figures 3.1 to 3.4: Projections of $\Omega_{D}$
and $\Omega_{AB}$ onto the planes $x_{A,1}p_{A,1}$, $x_{B,1}p_{B,1}$,
$x_{A,1}x_{B,1}$ and $p_{A,1}p_{B,1}$ for the case $a=0.7$, $b=1.8$ and
$\hbar=1$.} \vspace{0.3cm} } }

{\normalsize Finally, Figure 4 displays the possible values of $a$ and $b$ of
the enclosed ellipsoids $\Omega_{AB}$ for the example of Section
\ref{Section4.6}.  }

{\normalsize \vspace{1.2cm} \begin{figure}[h]
{\normalsize \center\includegraphics [scale=0.3] {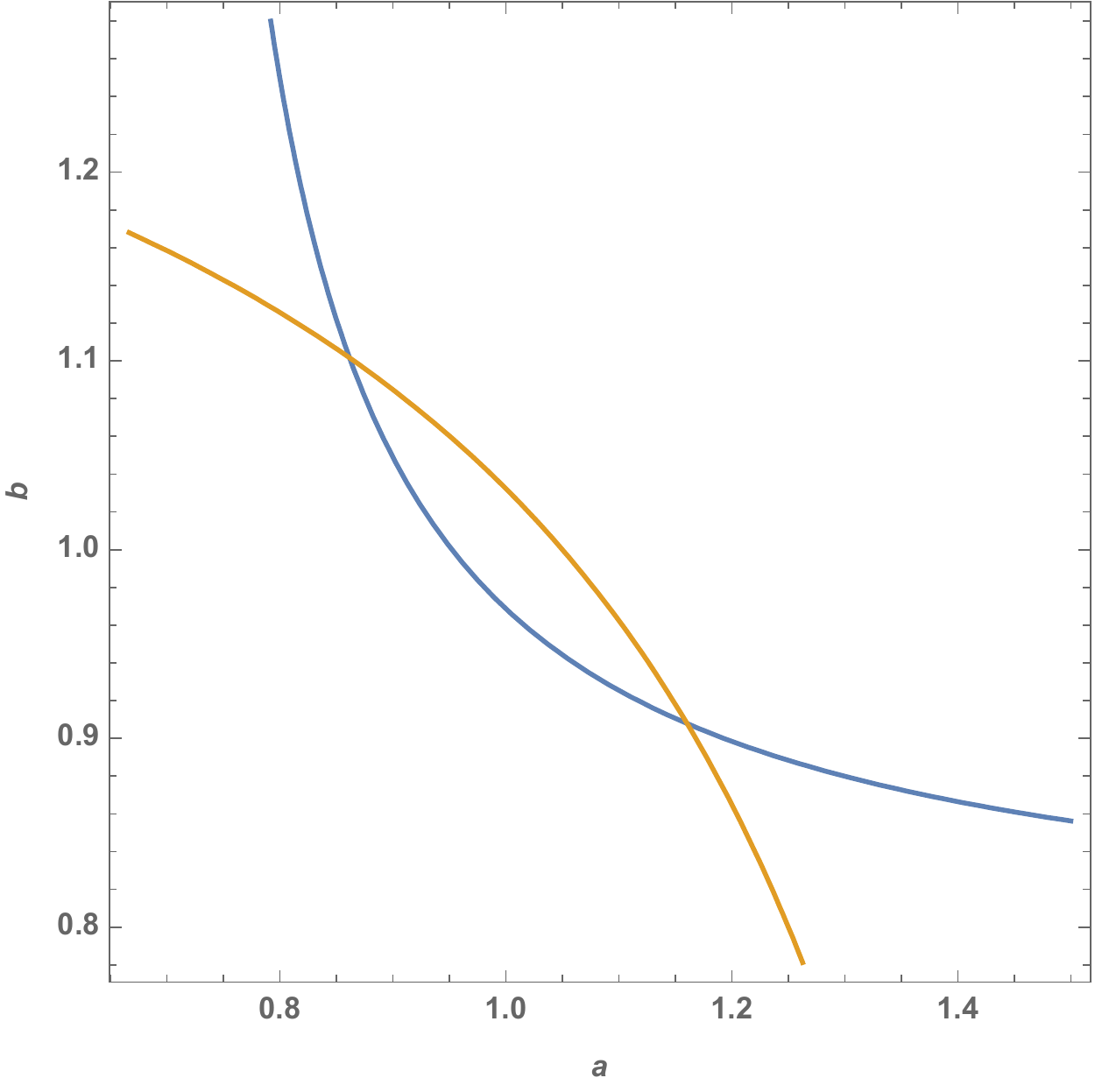}}\end{figure}\vspace{0cm} {\footnotesize {Figure 4: Numerical solutions of eq. (\ref{eqSep4.23}) for the case (\ref{eqImp29}). Each point between the two
curves corresponds to an ellipsoid $\Omega_{AB}$ that is enclosed in the
covariant ellipsoid $\Omega=\Omega_{D}$ associated to (\ref{eqImp29}).}
\vspace{0.1cm} } }

\section{Discussion{\protect\normalsize \label{RemarkImp1}}}

{\normalsize Since we have Theorem \ref{TheoremSep3} $\Rightarrow$ Theorem
\ref{TheoremSep2} $\Rightarrow$ Theorem \ref{TheoremSep1}, but the converse is
not valid, we conclude that only the criterion stated in Theorem
\ref{TheoremSep3} is a candidate for a necessary and sufficient condition for
separability of Gaussian states. Bearing this fact in mind, one could be
tempted to forget about Theorems \ref{TheoremSep1} and \ref{TheoremSep2}
altogether and keep only Theorem \ref{TheoremSep3} as a criterion. We
nevertheless feel that this hierarchy of criteria may be useful, because the
computational complexity increases from one criterion to the next. In
particular, it may not be easy to determine the optimal choice of numbers
$\varepsilon_{1},\cdots,\varepsilon_{n}>0$, to satisfy the condition of
Theorem \ref{TheoremSep3}. So, if one is able to prove separability using,
say, Theorem \ref{TheoremSep1}, then there is no need to apply the more
complicated Theorems \ref{TheoremSep2} and \ref{TheoremSep3}.  }

{\normalsize This situation is however in no way discouraging since it is
always easy to check whether a Gaussian is a good candidate to be a separable
state by using the very simple PPT criterion, which reduces to some trivial
manipulations of the covariance matrix.  }

\begin{acknowledgement}
{\normalsize Part of this research was done while Maurice de Gosson held the
Giovanni-Prodi-Lehrstuhl at the University of W\"{u}rzburg. The work of N.C. Dias and J.N. Prata was supported by the Portuguese Science
Foundation (FCT) grant PTDC/MAT-CAL/4334/2014.}
\end{acknowledgement}

****************************************************************

{\bf Author's addresses:}

\begin{itemize}
\item Nuno Costa Dias and Jo\~{a}o Nuno Prata: Grupo de F\'{\i}sica
Matem\'{a}tica, Departamento de Matem\'{a}tica, Faculdade de Ci\^{e}ncias,
Universidade de Lisboa, Campo Grande, Edif\'{\i}cio C6, 1749-016 Lisboa,
Portugal and Escola Superior N\'{a}utica Infante D. Henrique. Av. Eng.
Bonneville Franco, 2770-058 Pa\c{c}o d'Arcos, Portugal.

\item Maurice A. de Gosson: Universit\"{a}t Wien, Fakult\"{a}t f\"{u}r
Mathematik--NuHAG, Oskar-Morgenstern-Platz 1, 1090 Vienna, Austria.
\end{itemize}

****************************************************************

\begin{thebibliography}{99}                                                                                               %


\bibitem {abbo}{\normalsize A. Abbondandolo and R. Matveyev, How large is the
shadow of a symplectic ball?, \textit{J. Topol. Anal.} 5(01), 87--119 (2013)
}

\bibitem {abbobis}{\normalsize A. Abbondandolo and P. Majer, A non-squeezing
theorem for convex symplectic images of the Hilbert ball, \textit{Calc. Var.}
54, 1469--1506 (2015)  }

\bibitem {Bastiaans}{\normalsize M.J. Bastiaans, Wigner distribution function
and its application to first-order optics, \textit{JOSA} 69(12), 1710--1716
(1979)  }

\bibitem {Bochner}{\normalsize S. Bochner, \textit{Lectures on Fourier
integrals}, Princeton University Press, Princeton N.J., 1959  }

\bibitem {cogoni}{\normalsize E. Cordero, M. de Gosson, and F.\ Nicola, On the
Positivity of Trace Class Operators, \qquad arXiv:1706.06171 [math.FA]  }

\bibitem {digopra19}{\normalsize N. Dias, M. de Gosson, and J. Prata, On
Orthogonal Projections of Symplectic Balls, \qquad arXiv:1911.03763 [math.SG]
}

\bibitem {Duan}{\normalsize Lu-Ming Duan, G. Giedke, J.I. Cirac, and P. Zolle,
Inseparability Criterion for Continuous Variable Systems, \textit{Phys. Rev.
Lett}. 84(12), 2722--2725 (2000)  }

\bibitem {duwong}{\normalsize J. Du and M.W. Wong, A trace formula for Weyl
transforms, \textit{Approx. Theory. Appl.} (N.S.) 16(1), 41--45 (2000)  }

\bibitem {Folland}{\normalsize G.B. Folland, \textit{Harmonic Analysis in
Phase space}, Annals of Mathematics studies, Princeton University Press,
Princeton, N.J. (1989)  }

\bibitem {Giedke}{\normalsize G. Giedke, J. Eisert, J.I. Cirac, and
M.B.Plenio, Entanglement transformations of pure Gaussian states,
\textit{arXiv preprint quant-ph}/0301038 (2003)  }

\bibitem {Birk}{\normalsize M. de Gosson, \textit{Symplectic geometry and
quantum mechanics,} Vol. 166. Springer Science \& Business Media, 2006  }

\bibitem {FOOP}{\normalsize M. de Gosson, The Symplectic Camel and the
Uncertainty Principle: The Tip of an Iceberg? \textit{Found. Phys.} 39(2),
194--214 (2009)  }

\bibitem {Birkbis}{\normalsize M. de Gosson, \textit{Symplectic methods in
harmonic analysis and in mathematical physics,} Vol. 7. Springer Science \&
Business Media, 2011  }

\bibitem {blob}{\normalsize M. de Gosson, Quantum blobs. \textit{Found. Phys.}
43(4), 440--457\textbf{\ }(2013)  }

\bibitem {Wigner}{\normalsize M. de Gosson, \textit{The Wigner Transform},
World Scientific, Series: Advanced Texts in mathematics, 2017  }

\bibitem {QUANTA}{\normalsize M. de Gosson, Quantum Harmonic Analysis of the
Density Matrix, \textit{Quanta} 7 (2018)  }

\bibitem {goluPR}{\normalsize M. de Gosson and F. Luef, Symplectic Capacities
and the Geometry of Uncertainty: the Irruption of Symplectic Topology in
Classical and Quantum Mechanics. \textit{Phys. Rep.} 484, 131--179 (2009)  }

\bibitem {Gromov}{\normalsize M. Gromov, Pseudoholomorphic curves in
symplectic manifolds. Inv. Math. 82(2), 307--347 (1985)  }

\bibitem {horn}{\normalsize R.A. Horn and C.R. Johnson, \textit{Matrix
Analysis,} Cambridge University Press, 1990  }

\bibitem {HHH1}{\normalsize M. Horodecki, P. Horodecki, and R. Horodecki,
Separability of mixed states: necessary and sufficient conditions,
\textit{Phys. Lett. A }223, 1--8 (1996)  }

\bibitem {HHH2}{\normalsize M. Horodecki, P. Horodecki, and R. Horodecki,
Separability of n-particle mixed states: necessary and sufficient conditions
in terms of linear maps, \textit{Phys. Lett. A} 283(1-2), 1--7 (2001)  }

\bibitem {Kastler}{\normalsize D. Kastler, The $C^{\ast}$-Algebras of a Free
Boson Field, \textit{Commun. math. Phys.} 1, 14--48 (1965)  }

\bibitem {LouMiracle1}{\normalsize G. Loupias and S. Miracle-Sole, $C^{\ast}%
$-Alg\`{e}bres des syst\`{e}mes canoniques, I, \textit{Commun. math. Phys.} 2,
31--48 (1966)  }

\bibitem {LouMiracle2}{\normalsize G. Loupias and S. Miracle-Sole, $C^{\ast}%
$-Alg\`{e}bres des syst\`{e}mes canoniques, II, \textit{Ann. Inst. Henri
Poincar\'{e}} 6(1), 39--58 (1967)  }

\bibitem {Narcow89}{\normalsize F.J. Narcowich, Distributions of $\eta
$-positive type and applications, \textit{J. Math. Phys}., 30(11), 2565--2573
(1989)  }

\bibitem {Narcow90}{\normalsize F.J. Narcowich, Geometry and uncertainty,
\textit{J. Math. Phys.} 31(2),354--364 (1990)  }

\bibitem {Narconnell}{\normalsize F.J. Narcowich and R.F. O'Connell, Necessary
and sufficient conditions for a phase-space function to be a Wigner
distribution, \textit{Phys, Rev. A} 34(1), 1--6 (1986)  }

\bibitem {Peres}{\normalsize A. Peres, Separability criterion for density
matrices, \textit{Phys. Rev. Lett.} 77, 1413--1415 (1996)  }

\bibitem {Serafini}{\normalsize A. Serafini, Quantum Continuous Variables: A
Primer of Theoretical Methods, CRC Press, 2017  }

\bibitem {sh87}{\normalsize M.A. Shubin, \textit{Pseudodifferential Operators
and Spectral Theory}, Springer-Verlag, (1987)  }

\bibitem {Simon}{\normalsize R. Simon, Peres--Horodecki Separability Criterion
for Continuous Variable Systems, \textit{Phys. Rev. Lett.} 84(12), 2726--2729
(2000)  }

\bibitem {sisumu}{\normalsize R. Simon, E.C.G. Sudarshan, and N. Mukunda,
Gaussian-Wigner distributions in quantum mechanics and optics. \textit{Phys.
Rev. A} 36(8), 3868 (1987)  }

\bibitem {Tzon}{\normalsize Tzon-Tzer Lu and Sheng-Hua Shiou, Inverses of
$2\times2$ Block Matrices, \textit{Comput. Math. Appl.} 43, 119--129 (2002)  }

\bibitem {Werner}{\normalsize R. Werner, Quantum harmonic analysis on phase
space, \textit{J. Math.Phys.} 25(5), 1404--1411 (1984)  }

\bibitem {ww1}{\normalsize R.F. Werner and M.M. Wolf, Bound Entangled Gaussian
States, \textit{Phys. Rev. Lett.} 86(16), 3658 (2001)  }

\bibitem {zhang}{\normalsize F. Zhang, \textit{The Schur Complement and its
Applications}, Springer, Berlin, 2005  }
\end{thebibliography}
\end{document}